\documentclass[letterpaper, 10 pt, conference]{ieeeconf}  % Comment this line out if you need a4paper

\IEEEoverridecommandlockouts                              % This command is only needed if 
\usepackage[T1]{fontenc}
\usepackage{graphicx}
\usepackage{amsmath}
\usepackage{amssymb}
\usepackage{xcolor}
\usepackage{todonotes}
\usepackage{tikz}
\usepackage[ruled,vlined]{algorithm2e}
\usepackage{multirow}
\usepackage{subcaption}
\usepackage{lineno}

\usetikzlibrary{automata,positioning,decorations.markings,arrows,intersections,calc,shapes}
\usetikzlibrary{arrows.meta,fit,backgrounds}

	\newcommand{\ba}{\begin{array}}
		\newcommand{\ea}{\end{array}}
	\newcommand{\be}{\begin{equation}}
		\newcommand{\ee}{\end{equation}}
	\newcommand{\bea}{\begin{eqnarray}}
		\newcommand{\eea}{\end{eqnarray}}
	\newcommand{\bean}{\begin{eqnarray*}}
		\newcommand{\eean}{\end{eqnarray*}}
	\newcommand{\bc}{\begin{center}}
		\newcommand{\ec}{\end{center}}

    \newcommand{\pd}{\Delta}
	\newcommand{\pa}{s}

	\newcommand{\dis}{\textbf{Dis}}	
	\newcommand{\mo}{\mathcal{M}}

    \newcommand{\ap}{\mathcal{AP}}	
    \newcommand{\Aa}{\mathcal{A}}
    \newcommand{\alphabet}{\Sigma}

\tikzset{
  >=latex,node distance=2cm,on grid,auto, initial text=,
  box state/.style={draw,rectangle,minimum size=8mm,rounded corners},
  prob state/.style={draw,very thick,shape=circle,darkblue,minimum size=3mm,inner sep=0mm},
  every loop/.style={shorten >=0pt},
  accepting state/.style={double distance=1.2pt, outer sep = 0.6pt+\pgflinewidth},
  accepting dot/.style={above=-2.5pt,circle,fill,darkgreen,inner sep=2pt,radius=1pt},
  loop above/.append style={every loop/.append style={out=120, in=60, looseness=6}},
  loop below/.append style={every loop/.append style={out=300, in=240, looseness=6}},
  loop left/.append style={every loop/.append style={out=210, in=150, looseness=6}},
  loop right/.append style={every loop/.append style={out=30, in=330, looseness=6}}
}

 \newtheorem{assumption}{Assumption}
\usepackage{hyperref}
\usepackage{graphicx}
\usepackage{amssymb,amsmath,amsfonts}
\usepackage{mathtools}
\usepackage{ifthen,version}
\usepackage{tikz}
\usepackage{todonotes}
\usetikzlibrary{automata,positioning,decorations.markings,arrows,intersections,calc,shapes}

\pgfmathsetmacro{\arrLen}{1}        
\pgfmathsetlengthmacro{\HorizVarDistnc}{2cm}

\usepackage{cleveref}

\usepackage{makecell}
\usepackage{booktabs}

\colorlet{darkgreen}{green!40!black}
\colorlet{darkblue}{blue!60!black}
\colorlet{darkred}{red!50!black}
\colorlet{safecellcolor}{yellow!5}
\colorlet{goodcellcolor}{green!10}
\colorlet{badcellcolor}{blue!10}

\PassOptionsToPackage{usenames,dvipsnames}{xcolor}

\definecolor{dkgreen}{rgb}{0,0.6,0}
\definecolor{gray}{rgb}{0.5,0.5,0.5}
\definecolor{mauve}{rgb}{0.58,0,0.82}

\tikzset{
  >=latex,node distance=2cm,on grid,auto, initial text=,
  box state/.style={draw,rectangle,minimum size=8mm,rounded corners},
  prob state/.style={draw,very thick,shape=circle,darkblue,minimum size=3mm,inner sep=0mm},
  every loop/.style={shorten >=0pt},
  accepting state/.style={double distance=1.2pt, outer sep = 0.6pt+\pgflinewidth},
  accepting dot/.style={above=-2.5pt,circle,fill,darkgreen,inner sep=2pt,radius=1pt},
  loop above/.append style={every loop/.append style={out=120, in=60, looseness=6}},
  loop below/.append style={every loop/.append style={out=300, in=240, looseness=6}},
  loop left/.append style={every loop/.append style={out=210, in=150, looseness=6}},
  loop right/.append style={every loop/.append style={out=30, in=330, looseness=6}}
}

\newcommand{\set}[1]{\left\{ #1 \right\}}

\newcommand{\seq}[1]{\langle #1 \rangle}

\newcommand{\Ff}{\mathcal{F}}

\newcommand{\eE}{\mathbb E}

\newcommand{\Real}{\mathbb R}

\newtheorem{theorem}{Theorem}[section]
\newtheorem{lemma}[theorem]{Lemma}

\newtheorem{definition}{Definition}

\newtheorem{example}{Example}

\usepackage{wrapfig}
\usepackage{subcaption}  

\usepackage{thmtools,thm-restate}

\title{\LARGE \bf
Automata-Theoretic Verification of Interval Markov Decision Processes
}

\author{Sarvin Bahmani, Soumyajit Paul, Sven Schewe, Sadegh Soudjani and
 Ashutosh Trivedi% <-this % stops a space
\thanks{This research was supported in part by the following funding agencies: NSF under CAREER Award CCF-2146563, EPSRC through grant EP/X03688X/1, EIC 101070802, and ERC 101089047.}%
\thanks{A.\ Trivedi is a Royal Society Wolfson Visiting Fellow and gratefully acknowledges the support of the Wolfson Foundation and the Royal Society.}
\thanks{S. Bahmani, S. Paul, and S. Schewe are affiliated with the School of Computer Science and Informatics at the University of Liverpool. S. Soudjani is affiliated with MPI-SWS, Germany, and University of Birmingham, UK. A. Trivedi is affiliated with the Department of Computer Science at the University of Colorado Boulder.
Emails: $\mathtt{r.bahmani,soumyajit.paul,sven.schewe@liverpool.ac.uk}$, $\mathtt{sadegh@mpi\mbox{-}sws.org}$, $\mathtt{ashutosh.trivedi@colorado.edu}$} %
}

\begin{document}

\maketitle
\thispagestyle{empty}
\pagestyle{empty}

%%%%%%%%%%%%%%%%%%%%%%%%%%%%%%%%%%%%%%%%%%%%%%%%%%%%%%%%%%%%%%%%%%%%%%%%%%%%%%%%
\section{Abstract}
\label{sec:abstr}
\begin{abstract}
Interval Markov decision processes (IMDPs) provide a natural framework for modeling stochastic systems with uncertain transition probabilities, represented by probability intervals and resolved adversarially. Such uncertainty arises naturally, for example, when the transition model is learned from finite data or obtained through model-based reinforcement learning. In this paper, we study the automata-theoretic verification of IMDPs against rich temporal specifications, including all LTL specifications, by considering the broader class of $\omega$-regular objectives. We show that classical automata-theoretic verification techniques extend to IMDPs, but with a sharp distinction determined by the structure of the transition intervals. For stable IMDPs, where either the upper bound is zero or the lower bound is strictly positive, verification reduces to ordinary MDP analysis and can be carried out using the standard automata used in that setting (good-for-MDP automata). For unstable IMDPs, where intervals may include zero while the upper bound is strictly positive, verification becomes game-like and requires automata whose nondeterminism can be resolved on the fly (good-for-games automata). Building on these insights, we develop algorithms for verifying $\omega$-regular specifications over IMDPs and derive probabilistic guarantees when the interval model is learned from sampled data.
The resulting framework enables principled verification of stochastic systems under probabilistic model uncertainty, connecting automata-based verification with data-driven stochastic modeling.
\end{abstract}

\section{Introduction}
\label{sec:intro}
Stochastic systems in control, robotics, and learning are often modeled only approximately. In many settings, transition probabilities are not known exactly, but must be estimated from finite data or inferred through model-based reinforcement learning. This uncertainty makes formal reasoning difficult, particularly when one seeks to verify rich temporal properties rather than simple reachability or safety objectives.

A natural framework for capturing such uncertainty is provided by interval Markov decision processes (IMDPs), in which each transition probability is specified by an interval rather than a fixed value. The interval representation compactly describes a family of stochastic models and supports robust reasoning by interpreting the uncertainty adversarially. IMDPs have therefore emerged as a useful abstraction for data-driven stochastic systems, where exact transition laws are unavailable but confidence bounds or structural regularity assumptions allow one to bound transition probabilities.

To specify desired behaviors, the formal methods literature has long relied on temporal logics such as linear temporal logic (LTL) and, more generally, on the expressive class of $\omega$-regular objectives~\cite{Baier08,deAlfa98,Babiak15}. Automata-theoretic methods provide a powerful bridge between such specifications and algorithmic verification. This raises a natural question: how far do these classical verification techniques extend when the underlying system is an IMDP rather than an MDP?

In this work, we study the automata-theoretic verification and control of IMDPs against $\omega$-regular objectives, including all LTL specifications. Our goal is to understand when interval uncertainty preserves the classical MDP view of automata-based verification and when it instead induces a genuinely game-like setting which requires stronger automata that the product construction remains sound; good-for-games (GFG) automata suffice for this purpose.

\smallskip
\noindent \textbf{Interval MDPs.}
In data-driven stochastic control, IMDPs arise by abstracting an uncertain system into finitely many states and bounding transition probabilities by intervals derived from samples or structural assumptions on the dynamics. This yields a finite model that represents a family of stochastic systems and supports robust verification under adversarially resolved uncertainty.

Constructing an IMDP from data typically involves partitioning the continuous state space into finitely many regions and estimating possible transitions between regions under different control inputs.
% , as illustrated in Fig.~\ref{fig:imdp-abstraction}. 
Rather than relying on exact probability estimates, the abstraction bounds transition probabilities using intervals derived from statistical confidence bounds, regularity assumptions such as Lipschitz continuity, or other structural properties of the dynamics. Once an IMDP is available, one can apply probabilistic verification and policy-synthesis techniques to analyze temporal specifications under model uncertainty.

Recent work has developed methods for constructing IMDP abstractions from data. Badings et al.~\cite{badings2023robust} consider systems with known linear dynamics and additive noise of unknown distribution, yielding IMDPs with a fixed underlying graph structure. We refer to such IMDPs as \emph{stable}. Nazeri et al.~\cite{nazeri2025data} study fully unknown nonlinear systems with a known Lipschitz constant, yielding abstractions in which some transition intervals may include zero, so that the underlying graph may depend on how the uncertainty is resolved. We refer to such IMDPs as \emph{unstable}. Both works focus on reach-avoid objectives and compute optimal probabilities on the IMDP abstraction, thereby obtaining guarantees for probabilistic reachability and safety. In our paper, stability depends on the transition bounds, not on whether the underlying dynamics are linear or nonlinear. 

\smallskip
\noindent\textbf{Which automata are good for IMDPs?}
Automata-theoretic verification of $\omega$-regular objectives typically proceeds by translating the specification into an automaton and forming a product with the underlying stochastic model. A central difficulty is that expressive automata often involve nondeterminism. For such automata to be useful algorithmically, their nondeterministic choices must be resolvable in a way that preserves the satisfaction probability of the specification.
For ordinary MDPs, this difficulty is mitigated by the fact that the controller interacts only with stochastic dynamics rather than a strategic adversary. In this setting, \emph{good-for-MDP (GFM)} automata~\cite{Hahn20,Courco95,Brazdi14,Hahn15} suffice: their nondeterminism can be resolved without distorting satisfaction probabilities.

For IMDPs, transition uncertainty is resolved adversarially. Thus, although there is only one explicit controller, the interval uncertainty acts as an adversarial environment. Verification is therefore semantically closer to a game than to a standard MDP. In this setting, resolving automaton nondeterminism from the observed history alone may no longer preserve correctness. This is precisely where the stronger \emph{good-for-games (GFG)}~\cite{Henzin06} property becomes relevant: it requires that nondeterminism be resolved online, based only on the past, against an adversarial environment.

% For IMDPs, however, transition uncertainty is resolved adversarially. Although there is only one explicit controller, the uncertain transition probabilities act as an adversarial environment choosing values within the prescribed intervals. This makes the verification problem semantically closer to a game than to a standard MDP. In such a setting, resolving automaton nondeterminism based only on the history observed so far may no longer preserve correctness. This is precisely where the stronger \emph{good-for-games (GFG)}~\cite{Henzin06} property becomes relevant: it requires that nondeterminism can be resolved online, based only on the past, against an adversarial environment.

This observation reveals the central challenge addressed in this paper: interval uncertainty changes the semantics of automata-based verification. Automata that are sound for MDPs need not remain sound for IMDPs, but where they are, it allows for using significantly faster algorithms.
The precise boundary depends on whether the uncertainty preserves or can alter the connectivity structure of the model.

\smallskip
\noindent\textbf{Contributions.}
We study verification and control of IMDPs against $\omega$-regular objectives. Our main contributions are:
\begin{enumerate}
    \item We identify the automata-theoretic boundary for IMDP verification. In particular, we show that GFM automata suffice for stable IMDPs, where interval bounds exclude zero so that the adversarial player cannot alter the connectivity structure, whereas unstable IMDPs require the stronger GFG property.
    \item We develop algorithms for computing worst-case optimal policies for IMDPs with $\omega$-regular objectives, combining automata-theoretic constructions with insights from probabilistic parity games~\cite{Chatte03,Chatte04}.
    \item We implement two solution procedures: (i) a value iteration method that introduces slight imprecision because of the contraction used for upper bounds, and (ii) a strategy-improvement method that is slightly more expensive but exact. Unlike previous approaches, our method iterates over nature strategies, yielding small linear programs instead of exponential-size resolver programs. We evaluate the approach on large-scale case studies with diverse $\omega$-regular objectives, demonstrating scalability and practical applicability.
\end{enumerate}
Proofs and more details can be found in the appendix. %the full version~\cite{imdpverify}.

\section{Preliminaries}
\label{sec:preliminaries}
A \emph{probability distribution} over a  set $X$ is a function $d \colon X {\to} [0, 1]$ such that $\sum_{x \in X} d(x) = 1$.  Let
$\dis(X)$ denote the set of all distributions over $X$. 
% We
% say a distribution ${d \in \dis(X)}$ is a \emph{point distribution}
% if $d(x) {=} 1$ for some $x \in X$.  
% For $d \in \dis(X)$ we write
% $\supp(d)$ for $\set{x \in X : d(x) > 0}$. 
% A \emph{sequence} is an ordered list of elements $\seq{x_1, x_2, x_3, \ldots}$.
For a finite or infinite sequence $\pa = \seq{x_1, x_2, \ldots}$,
we write $\pa(k)$ to denote its $k$-th element.
For a finite sequence $s = \seq{x_1, \ldots, x_n}$ and an element $y$, we write $s \cdot y$ for the sequence $\seq{x_1, \ldots, x_n, y}$.
For a finite sequence $\pa = \seq{x_1, x_2, \ldots, x_n}$, we write $\textsf{last}(\pa)$ for the final element $x_n$ of $\pa$.

\begin{definition}[IMDP]\label{def_imdp}
An Interval MDP (IMDP) $\mo$ is a tuple $(X,x_0, A, \overline{P}, \underline{P},\ap,L)$ where 
    $X$ is a finite set of states; $x_0\in X$ is the initial state;
    $A$ is a finite set of actions; 
    $\overline{P}:X {\times} A {\times} X \to [0,1]$ and 
    $\underline{P}:X {\times} A {\times} X \to [0,1]$ are upper and lower transition probability bound functions;
    $\ap$ is a finite set of atomic propositions and 
    $L:X\to 2^\ap$ is a labeling function.
Furthermore, to ensure $\overline{P}$ and $\underline{P}$ admit well-formed transition probabilities, we require that, for any $x, x'{\in} X$ and $a {\in} A$, we have $\underline{P}(x,a,x') \leq \overline{P}(x,a,x')$, and
\[
    \sum_{x'\in X} \underline{P}(x,a,x') \leq 1 \leq  \sum_{x'\in X} \overline{P}(x,a,x').
\]  
When $\overline{P} = \underline{P}$, the IMDP becomes an MDP.
\end{definition}
Given a state $x\in X$ and input $a \in A$, we denote by $\pd(x,a)$ the set of feasible distribution over $X$, i.e., 
\begin{multline}
\label{feasible-distribution}
    \pd(x,a)=
    \Big\{ p\in \dis(X): \sum_{x' \in X} p(x'|x,a){=}1   \text{ and }\\
 \underline{P}(x,a,x')\leq p(x'|x,a) \leq \overline{P}(x,a,x')  \text{ for all } x'{\in} X\Big\}.
\end{multline}
Note that $\pd(x,a)$ is a polytope in $\mathbb R^{|X|}$. 
Given a state $x \in X$, we denote the set of all actions that are enabled at state $x$ by $\Gamma(x)=\{a \in A \mid \text{ there exists } x' \in X, \text{ such that } \overline{P}(x,a,x') >0\}$.
A finite path of an IMDP $\mo = (X,x_0, A, \overline{P}, \underline{P},\ap,L)$ is a finite sequence of states $\pa = \seq{x_1, x_2, \ldots, x_n}$ where $x_1 = x_0$ and, for all $1 < i \leq n $, there exists $a \in \Gamma(x_{i-1})$ such that $\overline{P}(x_{i-1},a,x_i) > 0$.  
An infinite path is an infinite such sequence. 
We denote by $\textsf{Path}^*(\mo)$ and $\textsf{Path}^\omega(\mo)$ the set of all finite paths and infinite paths in $\mo$, respectively.
Let $\textsf{Path}(\mo)=\textsf{Path}^*(\mo) \cup \textsf{Path}^\omega(\mo)$ be the set of all paths.

\smallskip
\noindent
\textbf{The policy and the nature.}
A \emph{policy} is a function
$\pi : \textsf{Path}^*(\mo) \to \dis(A)$,
which assigns a distribution over action set to each finite path. For every finite path $s$,
$\pi(s,a)=0$ whenever $a\notin\Gamma(\mathrm{last}(s))$.
A \emph{nature} is a function
$\theta : \textsf{Path}^*(\mo) \times A \to \dis(X)$,
such that, for every finite path $s \in \textsf{Path}^*(\mo)$ and action $a \in A$, the distribution
$\theta(s, a)$ is a feasible successor distribution in $\pd(\textsf{last}(s), a)$.
Let $\Pi(\mo)$ and $\Theta(\mo)$ be the sets of all policies and natures for~$\mo$.

For $(\pi, \theta) \in \Pi(\mo) \times \Theta(\mo)$, let
$\textsf{Path}^{\pi,\theta}_\mo$ denote the set of infinite runs of $\mo$ starting at $x_0$ that are consistent with $(\pi,\theta)$.
Given a finite path $\pa = \seq{x_1, x_2, \ldots, x_n} \in \textsf{Path}^*(\mo)$, the \emph{cylinder set} generated by $\pa$ is
\[
\textsf{Cyl}(\pa)
  = \bigl\{
      \rho \in \textsf{Path}^\omega(\mo)
      \,\big|\,
      \rho_i = \pa_i \text{ for } 1 \le i \le n
    \bigr\}.
\]
Let $\Ff_\mo^{\pi,\theta}$ be the $\sigma$-algebra generated by cylinder sets. Then the behavior of $\mo$ under $(\pi,\theta)$ is described by the probability space
$(\textsf{Path}^{\pi,\theta}_\mo, \Ff_\mo^{\pi,\theta}, \Pr^{\pi,\theta}_\mo)$,
where the measure $\Pr^{\pi,\theta}_\mo$ is uniquely determined by the Ionescu--Tulcea extension theorem.
Given some $f \colon \textsf{Path}^\omega(\mo) \to \Real$, we write $\eE_\mo^{\pi,\theta}[f]$ for its expectation with respect to this probability space.

\noindent\textbf{Parity IMDPs.}
When analyzing IMDPs with parity objectives, we consider \emph{parity IMDPs}, which are tuples $\mo=(X,x_0, A, \overline{P}, \underline{P},\alpha)$, where $(X,x_0,A,\overline{P}, \underline{P})$ is the interval transition structure of an IMDP as in Definition~\ref{def_imdp}, and $\alpha:X \rightarrow \mathbb N$ is a priority function that maps states to natural numbers, called \emph{colors}.
For infinite paths in the IMDP, we call the highest color that occurs infinitely often \emph{dominating}, and the path \emph{accepting} if the dominating color is even.
For a parity IMDP $\mo$, a policy $\pi$ and a nature $\theta$,
we let
\[
\textsf{Pr}_{\mo}^{\pi,\theta} = \Pr_{\mo}^{\pi,\theta} \set{ s \in \textsf{Path}^{\pi,\theta}_\mo \::\: s \text{ is an accepting path}}.
\]
Let $\textsf{Pr}_{\mo}^{\pi} = \inf_\theta \textsf{Pr}_{\mo}^{\pi,\theta}$ be the (guaranteed) probability that $\mo$ creates an accepting path (regardless of nature) under policy $\pi$, and
    $\textsf{Pr}_{\mo}= \sup_\pi \textsf{Pr}_{\mo}^{\pi}$ be the (achievable) probability that $\mo$ creates an accepting path for any policy $\pi$.

\noindent\textbf{Stochastic parity games.}
A stochastic parity game (SPG) $\mathcal G$ is a tuple $(X_e,X_o,X_r,x_0, P,E,\alpha)$ where  $X=X_e \cup X_o \cup X_r$ is a finite set of states, partitioned into three pairwise disjoint sets: states $X_e$ owned by the maximizer, states $X_o$ owned by the minimizer, and random states $X_r$; $x_0\in X$ is the initial state; $\alpha:X \rightarrow \mathbb N$ is a coloring function; 
$P: X_r \times X \rightarrow [0,1]$ is a probability function, so that $\sum_{x\in X} P(x',x) = 1$ holds for all $x'\in X_r$; and $E \subseteq (X_e \cup X_o) \times X$ is a transition relation that maps at least one successor to each state in $X_e \cup X_o$.
SPGs are played by placing a token on $x_0$. When the token is on a state in $X_e$, the maximizer chooses the 
successor state from $E$, when it is on a state in $X_o$, the minimizer chooses, and for states in $X_r$, the successor state is sampled with the probabilities from $P$.
The objective of the maximizer (minimizer) is to maximize (minimize) the chance that the dominating color in the ensuing play is even.

\noindent\textbf{Parity and B\"uchi Automata.}
A deterministic parity automaton $\mathcal D$ is a tuple $\mathcal D = (\Sigma, Q, q_0, \delta_{\mathcal N}, \alpha)$, where $\Sigma=2^\ap$ and $\ap$ is a set of atomic propositions (Definition~\ref{def_imdp}), $Q$ is a finite set of states, $q_0 \in Q$ is the initial state, $\delta_{\mathcal N} \colon Q \times \Sigma \to Q$ is the transition function, and $\alpha \colon Q \to \mathbb N$ is the coloring function.
A deterministic parity automaton reads an infinite word over the alphabet $\Sigma$ and induces a unique infinite run over $Q$ starting from $q_0$. A run is accepting if the highest color that appears infinitely often along the run is even. The set of words whose runs are accepting is the language of the automaton.

A nondeterministic parity automaton $\mathcal N$ is a tuple $\mathcal N = (\Sigma, Q, Q_0, \Delta_{\mathcal N}, \alpha)$ defined as above, except that the set of initial states $Q_0$ need not be a singleton, and the transition function $\Delta_{\mathcal N} : Q \times \Sigma \to 2^Q$ maps each state and input letter to a set of successor states.
As a result, a nondeterministic automaton typically admits multiple runs on a given word, and its language consists of all words for which at least one run is accepting.
Nondeterministic automata $\mathcal N$ with the property that resolving the nondeterminism in $\mathcal N$ yields, for every finite MDP $\mo$, the same value as computing an accepting run in the parity MDP obtained from the product of $\mo$ and $\mathcal N$ are called \emph{good-for-MDPs}.
Parity conditions for which the image of $\alpha$ is ${1,2}$ are called \emph{B\"uchi} conditions.

\section{What's Good for IMDPs?}
\label{sec:goodForIMDP}
We develop automata-theoretic algorithms for analyzing IMDPs, which raises the question of which automata are appropriate for this setting.
We show that this question turns crucially on whether transition intervals include zero, since zero-valued lower bounds allow the environment to suppress certain transitions altogether. This feature typically depends on the construction of the IMDP: some abstraction methods yield such intervals, whereas others avoid them.

\begin{definition}[SIMDPs]
We call an IMDP \emph{stable} if, for all $x,x'\in X$ and $a \in A$, either $\underline{P}(x,a,x') >0$ or $\overline{P}(x,a,x')=0$ holds.  
Intuitively, stability ensures that every transition allowed by the abstraction remains possible under all admissible probability selections of the environment.
\end{definition}

In this paper, we focus primarily on stable IMDPs (SIMDPs). For this class, we show that qualitative analysis coincides with that of ordinary MDPs, and therefore good-for-MDP B\"uchi automata~\cite{Hahn20} suffice to capture $\omega$-regular objectives. Moreover, quantitative analysis reduces to reachability, allowing the use of standard techniques for reachability games, such as value iteration and strategy improvement. This reduction is specific to stable IMDPs and does not hold for general IMDPs.

Before turning to SIMDPs, we first outline an approach for analyzing general IMDPs. The main idea is to construct the synchronous product of an IMDP with a suitable automaton whose memory acts as a witness for acceptance. For general IMDPs, this automaton must be a deterministic parity automaton, or more generally a GFG parity automaton~\cite{Henzin06}, recognizing the target $\omega$-regular language on $2^{\ap}$.

To analyze an IMDP $\mo = (X,x_0,A,\overline{P},\underline{P},\ap,L)$ against a specification given by a deterministic parity automaton $\mathcal{D} = (\alphabet,Q,q_0,\delta_{\mathcal N},\alpha)$, we construct a stochastic parity game $\mathcal{G}$ as follows:
\begin{itemize}
    \item \textbf{Maximizer states:} The states are $X \times Q$, with initial state $(x_0,q_0)$.
    \item \textbf{Maximizer choices:} From a state $(x,q)$, the maximizer chooses an action $a \in \Gamma(x)$ and moves to the corresponding minimizer state $(x,q,a)$.

\item \textbf{Minimizer states:} These are states $(x,q,a) \in X \times Q \times A$ with $a \in \Gamma(x)$. From such a state, the minimizer selects a corner $p$ of the polytope $\pd(x,a)$, yielding a random state $(x,q,a,p)$.\footnote{A corner is a point that cannot be written as a nontrivial convex combination of two distinct points of the polytope. Any probabilistic choice $p\in\pd(x,a)$ can be expressed
as a convex combination of the vertices of $\pd(x,a)$.}

\item \textbf{Random states:} From $(x,q,a,p)$, the game moves to a successor $(x',q')$ with probability $p(x' \mid x,a)$, where $q' = \delta_{\mathcal N}(q,L(x))$.

\item \textbf{Colors:} The coloring $\alpha'$ is inherited from $\mathcal{D}$:
\[
\alpha'(x,q) {=} \alpha(q), \;
\alpha'(x,q,a) {=} \alpha(q), \;
\alpha'(x,q,a,p) {=} \alpha(q).
\]
\end{itemize}
Alternatively, we may construct a parity IMDP $\mathcal{P} = \bigl(X \times Q,(x_0,q_0),A,\overline{P}',\underline{P}',\alpha'\bigr)$, in which the automaton component is updated deterministically according to $q'=\delta_{\mathcal N}(q,L(x))$, while the system component evolves as in $\mo$. Hence, the transition bounds $\overline{P}'$ and $\underline{P}'$ depend only on the $X$-component, and the coloring is given by $\alpha'(x,q)=\alpha(q)$.

% Alternatively, we may form a parity IMDP $\mathcal{P} = \big(X \times Q,(x_0,q_0),A,\overline{P}',\underline{P}',\alpha'\big)$,
% where the second component of each state $(x,q)$ is updated deterministically according to the automaton, i.e., $q' = \delta_{\mathcal N}(q,L(x))$, while the first component evolves as in $\mo$. Thus, the transition probabilities $\overline{P}'$ and $\underline{P}'$ depend only on the $X$-component. The coloring is similar : $\alpha'(x,q) = \alpha(q)$.

Note that $\mathcal{G}$ may be exponentially larger than $\mo$. This is because $\pd(x,a)$ may have exponentially many vertices, each inducing a random state, even when $\mo$ is an SIMDP. For example, if the support of $\pd(x,a)$ has size $2k+1$ for some $k \in \mathbb{N}$ and each interval is
$\big[\frac{1}{2k+2},\,\frac{3}{2k+2}\big]$,
then $\pd(x,a)$ has $(2k+1)\binom{2k}{k}$ vertices. Each vertex corresponds to a distribution in which one successor has probability $\frac{1}{k+1}$, exactly $k$ successors have probability $\frac{1}{2k+2}$, and the remaining $k$ have probability $\frac{3}{2k+2}$.

\begin{lemma}
For every policy $\pi$ and every nature $\theta$,
the acceptance probability
$\textsf{Pr}^{\pi,\theta}_{\mathcal{P}}$
coincides with the probability that $\mo$ generates
a path whose label trace is accepted by $\mathcal{D}$.
\end{lemma}

\begin{proof}
By construction of the product, $\mathcal{P}$ and $\mo$ induce the same probability space over paths, generated by the same choices of $\pi$ and $\theta$.
Moreover, along each path, the automaton component is updated exactly according to $\delta_{\mathcal N}$, so the color sequence in $\mathcal{P}$ matches the acceptance condition of $\mathcal{D}$ on the corresponding label trace of $\mo$. Hence, a path is accepting in $\mathcal{P}$ if and only if its trace in $\mo$ is accepted by $\mathcal{D}$. 
\end{proof}

% \noindent Consequently, $\mathcal{P}$ preserves both the correct probabilities and the optimal policies.

% \begin{lemma}
% The maximizer has the same pure positional optimal policies, achieving the same winning probabilities, in $\mathcal{P}$ and in $\mathcal{G}$.
% \end{lemma}

\begin{lemma}
The maximizer's pure positional optimal policies in $\mathcal{P}$ and $\mathcal{G}$ correspond under the
natural identification of maximizer states and actions, and achieve the same winning probabilities.
\end{lemma}

\begin{proof}
Finite SPGs admit optimal pure positional strategies~\cite{DBLP:journals/jcss/AlfaroM04,DBLP:conf/fossacs/Zielonka04}. Let $p$ denote the probability of satisfying the parity objective in $\mathcal G$ when both players follow optimal pure positional policies
$\pi : X \times Q \rightarrow A$ and
$\theta : X \times Q \times A \rightarrow (X \times Q \rightarrow [0,1])$.
Note that each transition in $\mathcal{P}$ is expanded into three steps in $\mathcal{G}$, but this affects only the encoding of paths, not the induced values or optimal choices.

We now show that these policies are also optimal in $\mathcal P$ and yield the same probability $p$.
\begin{enumerate}
    \item
    Suppose for contradiction that there exists a policy $\pi'$ in $\mathcal P$ that achieves a higher probability $p' > p$ against nature $\theta$.
    Then $\pi'$ induces a corresponding policy against $\theta$ in $\mathcal G$, yielding the same winning probability $p' > p$,
    contradicting the optimality of $\pi$.

    \item
    Conversely, suppose for contradiction that there exists a nature strategy $\theta'$ in $\mathcal P$ that achieves a lower winning probability $p' < p$ against $\pi$.
    Then $\theta'$ induces a corresponding strategy against $\pi$ in $\mathcal G$, again yielding probability $p' < p$,
    contradicting the optimality of $\theta$.
\end{enumerate}
Therefore, $\pi$ and $\theta$ remain optimal in $\mathcal P$, and the winning probability is the same in $\mathcal P$ and $\mathcal G$.
\end{proof}

Efficient algorithms for both qualitative~\cite{Hahn16} and quantitative~\cite{DBLP:conf/vmcai/HahnST017} analysis of stochastic parity games are available. 
In particular, the qualitative analysis can, in principle, be solved via a gadget-based reduction to non-stochastic parity games~\cite{Chatte04}, which can then be handled in quasi-polynomial time~\cite{DBLP:journals/siamcomp/CaludeJKLS22,DBLP:journals/lmcs/LehtinenPSW22,DBLP:journals/fcomp/DellErbaS22}.
For an IMDP $\mo$, the \emph{qualitative analysis} problem is determining the set of states $x \in X$ for which $\mathsf{Pr}_{\mo_x}=0$ or $\mathsf{Pr}_{\mo_x}=1$ holds (where $\mo_x=(X,x,A,\overline{P},\underline{P},\alpha)$).

\begin{theorem}
The qualitative analysis of IMDPs is EXPTIME-complete when the property is given as a nondeterministic B\"uchi automaton, and 2EXPTIME-complete when the property is given as an LTL formula.
\end{theorem}

\begin{proof}
For the upper bound, note that an NBA $\mathcal{N}$ with $s$ states can be determinized into a parity automaton $\mathcal{D}$ with an exponential number of states ($s!^2$) and with linearly many priorities ($2s{+}1$).
The resulting stochastic parity game $\mathcal{G}$ has the same priorities as $\mathcal{D}$ and a number of states that is exponential in both $\mathcal{D}$ and $\mo$.
Qualitative analysis of $\mathcal{G}$ can then be performed in time exponential in the sizes of $\mo$ and $\mathcal{N}$ (doubly exponential in the size of an LTL formula $\varphi$).
Matching lower bounds follow from the complexity of MDP analysis~\cite{Courco95}.
%\qed
\end{proof}

%%%%%
\subsection{Quantitative Analysis}
We now establish three upper bounds for the \emph{quantitative analysis} problem, i.e., computing the maximal satisfaction probability of an IMDP $\mo$ against an NBA~$\mathcal{N}$.

\begin{theorem}
The quantitative analysis of IMDPs against NBA specifications satisfies the following:
\begin{enumerate}
    \item it is in $\textsf{UEXPTIME} \cap \textsf{coUEXPTIME}$, and
    \item it can be solved in co-nondeterministic time which is polynomial in the size of the IMDP and exponential in the size of the NBA.
\end{enumerate}
\end{theorem}

\begin{proof}
For (1), we use that stochastic parity games belong to a class of stochastic games that are polynomial-time interreducible~\cite{DBLP:conf/isaac/AnderssonM09} and that membership of parity games in $\textsf{UP} \cap \textsf{coUP}$ is well established (see, e.g.,~\cite{Jurdzi98} for parity games, which admit a unique value).
For (2), we exploit that both players have optimal pure positional strategies.
To refute an inequality of the form $\mathsf{Pr}_\mo \ge p_0$, it suffices to nondeterministically guess a pure positional minimizer strategy. 
This guess can be made incrementally while constructing the game, guaranteeing that each minimizer state has a single successor.
The resulting structure is then an ordinary parity MDP of size polynomial in~$\mo$ and exponential in~$\mathcal{N}$ (via the determinization of~$\mathcal{N}$ into a parity automaton~$\mathcal{D}$).
Solving such parity MDPs can be done in polynomial time.
%\qed
\end{proof}

As these bounds are expressed in terms of $\mathcal{N}$, the corresponding upper bounds for LTL follow directly from translating an LTL formula into an equivalent, but potentially exponentially larger, NBA~$\mathcal{N}$.

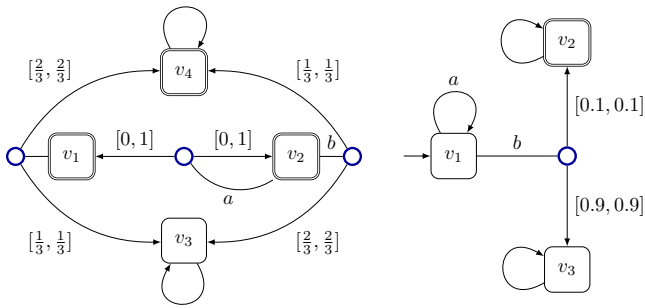
\begin{figure}[t]
\centering
\begin{subfigure}[h]{0.27\textwidth}
    \centering
    \resizebox{\linewidth}{!}{\begin{tikzpicture}[every node/.style={transform shape}]
\node[prob state] (s1) {};
\node[box state,accepting] (v1) [right=1cm of s1] {$v_1$};
\node[prob state] (s2) [right=2cm of v1] {};
\node[box state,accepting] (v2) [right=of s2] {$v_2$};
\node[prob state] (s3) [right=1cm of v2] {};

\node[box state] (v3) [below=1.5cm of s2] {$v_3$};
\node[box state,accepting] (v4) [above=1.5cm of s2] {$v_4$};

% --- Transitions ---
\draw[->, bend right] (s1) to node[below left] {$[\frac{1}{3}, \frac{1}{3}]$} (v3);

\draw[->, bend left] (s1) to node[above left]{$[\frac{2}{3},\frac{2}{3}]$} (v4);

\draw[->, bend left] (s3) to node[below right]{$[\frac{2}{3},\frac{2}{3}]$} (v3);
\draw[->, bend right] (s3) to  node[above right]{$[\frac{1}{3},\frac{1}{3}]$} (v4);

\draw[->] (s2) -- node[above]{$[0,1]$}(v1);
\draw[->] (s2) -- node[above]{$[0,1]$}(v2);
\draw[-] (v2) -- node[above]{$b$} (s3);
\draw[-] (v1) -- (s1);
\draw[-] (v2) to [bend right=-45] node[below]{$a$} (s2);
\draw[->] (v3) edge[loop below] node {} (v3);
\draw[->] (v4) edge[loop above] node {} (v4);
\end{tikzpicture}}
\end{subfigure}
\hfill
\begin{subfigure}[h]{0.2\textwidth}
    \centering
    \resizebox{\linewidth}{!}{%!TEX root = ../main.tex

\begin{tikzpicture}[scale=1, every node/.style={transform shape}]
\node[initial,box state] (v1) {$v_1$};
\node[prob state] (s1) [right=of v1] {};
\node[box state,accepting] (v2) [above=of s1] {$v_2$};

\node[box state] (v3) [below=of s1] {$v_3$};

\draw[->] (s1) -- node[right]{$[0.1, 0.1]$}(v2);
\draw[->] (s1) -- node[right]{$[0.9, 0.9]$}(v3);
\draw[-] (v1) -- node[above]{$b$} (s1);

\draw[->] (v1) edge[loop above] node {$a$} (v1);
\draw[->] (v2) edge[loop left] node {} (v2);
\draw[->] (v3) edge[loop left] node {} (v3);
\end{tikzpicture}}
\end{subfigure}
\caption{(Left) B\"uchi IMDP where the maximizer can win from $v_2$ with probability $\frac{2}{3}$ by always playing $a$, but reaches the almost-sure winning region only with probability $\frac{1}{3}$. 
(Right) Reachability IMDP, where interval value iteration does not terminate without contraction of the upper value. Upper value remains strictly greater than lower value by a fixed margin after few steps.}
%\vspace{-2em}
\label{fig-imdps}
\end{figure}

% \begin{figure}[t]
% \centering
%     \input{Figures/fig-imdp01}\hfill
%     \input{Figures/fig-imdp02}
% \caption{(Left) B\"uchi IMDP where the maximizer can win from $v_2$ with probability $\frac{2}{3}$ by always playing $a$, but reaches the almost-sure winning region only with probability $\frac{1}{3}$. 
% (Right) Reachability IMDP, where interval value iteration does not terminate without contraction of the upper value. Upper value remains strictly greater than lower value by a fixed margin after few steps.}
% %\vspace{-2em}
% \label{fig-imdps}
% \end{figure}

We note that, while we can in principle reduce the problem to a reachability game using the construction from \cite{DBLP:conf/isaac/AnderssonM09}, this reduction introduces gadgets with extremely small probabilities. 
In particular, these are \emph{not} reductions to the reachability of the almost sure (a.s.) winning region from the qualitative analysis, and such reductions would not yield correct results.
This is shown in Figure~\ref{fig-imdps}(left), which also shows why naive value iteration fails for these games.

\begin{example}\textit{(Reachability of a.s.\ Winning Regions vs.\ Winning)}: 
In Figure~\ref{fig-imdps}(left), the a.s.\ winning region is $\{v_4\}$.
If the maximizer wants to maximize her chances of reaching it from $v_2$, her best move is to choose action $b$, because the minimizer can push her back to state $v_2$ if she plays $a$.
Playing $b$ reaches $v_2$ and wins with a chance of $\frac{1}{3}$, while it reaches reaching $v_3$ and loses with a chance of $\frac{2}{3}$.
The best strategy for the maximizer, however, is to always play $a$ from state $v_2$. While this does not guarantee reach her winning (almost) sure winning as the minimizer can always play back to $v_2$, the minimizer loses if he does this forever, and always playing $a$ guarantees a chance of winning of at least $\frac{2}{3}$ to the maximizer.
This shows that maximizing the chance of reaching an a.s.\ winning region is not the same as maximizing the chance of winning.
\end{example}

\subsection{Good-for-IMDP automata}
We have presented the product construction using deterministic automata only for simplicity of exposition. The same construction applies equally well to good-for-games automata (GFG)~\cite{Henzin06,DBLP:conf/fsttcs/BokerKLS20} that recognize the same language, with the maximizer choosing an automaton successor $q'\in\Delta_{\mathcal N}(q,L(x))$ together with the action $a$,
before the minimizer (nature) chooses the distribution of transition probabilities. Conversely, labeled stochastic games can be encoded as IMDPs in a straightforward manner. It follows that the class of automata suitable for IMDP analysis cannot be more general than the class of automata suitable for games.
An automaton is good for a class of IMDPs if, for every IMDP in that class, the best guaranteed probability of generating a word in the automaton’s language coincides with the best guaranteed probability of acceptance in the product, where the maximizer resolves the automaton’s nondeterminism.
GFM ensure equality between the two probabilities for stable IMDPs, while GFG ensure this equality for general IMDPs.

\iffalse
\paragraph{\bf Qualitative analysis of SIMDPs}
The qualitative analysis of SIMDPs is the same as for MDPs.

\begin{theorem}
Given a nondeterministic B\"uchi automaton $\mathcal N$ and an SIMDP $\mo$, the qualitative analysis of $\mo$ can be done in time polynomial in $\mo$ and exponential in $\mathcal N$, and in time polynomial in $\mo$ and $\mathcal N$ if $\mathcal N$ is good for MDPs.
\end{theorem}

[There is some completeness formulation in \cite{Courco95}, looking for the right formulation.]

\begin{proof}
We first observe that, for the qualitative analysis, only the support of the probability distribution matters, not the probabilities themselves, cf.~\cite{Chatte04,Hahn16}.
As a consequence, the choice of the environment in the game are irrelevant and we can instead work with any distribution from $\pd(x,a)$, and thus with a consistent one, which can therefore be chosen arbitrarily and consistently for all states $(x,q,a)$ (i.e.\ independent of $q$).

This brings us back to the qualitative analysis of ordinary MDP. For the complexity and hardness for ordinary NBA see \cite{Courco95,Hahn15}, for good-for-MDP automata see \cite{Hahn20}.
\end{proof}
\fi

\section{Value Iteration and Strategy Improvement}
\label{sec:algo}
We propose two algorithms for the verification of stable IMDPs against $\omega$-regular properties. We have implemented these algorithms and report experimental results in \Cref{sec:experiments}.
The first algorithm is a value iteration algorithm, while the second one is a strategy improvement variant of. The (interval) value iteration algorithm keeps track of upper and lower values, and deploys a contraction factor for the upper value. The strategy improvement differs from the existing strategy improvement algorithms for IMDPs, and iterates over nature rather than policies.  The first step in both the algorithms is a qualitative analysis of the SIMDP. This procedure is identical to MDP qualitative analysis. 

\paragraph{Qualitative analysis}
The qualitative analysis is performed by ignoring the exact probability values, since in SIMDP all probabilities are positive. Recall that this is essentially computing the almost-sure winning regions and sure losing regions of some reachability game, that we describe shortly.
Given an SIMDP $\mo$ and an automaton $\Aa = (\alphabet, Q,q_0,\Delta_{\mathcal N},\alpha)$, we compute the product $\mo \times \Aa$, with states $X \times Q$. 
 On the underlying MDP structure comprising of all positive probabilities in $\mo$,    
we first identify the accepting maximal end components of $\mo \times \Aa$ and then collapse them into a target (winning sink)  $T$, thus forming a reachability game.
We compute the almost sure winning region $W^+$, from where $T$ can be reached almost surely.
The sure losing region $W^-$ are those states, from which the targets cannot be reached under any policy and collapse them into a losing sink. These computations determine initial values for our value iteration step.

\begin{restatable}{theorem}{qualitativeSIMDP}\label{thm:qualitative-SIMDP}
Given an SIMDP $\mo$ and automaton $\mathcal N$, the qualitative analysis of $\mo$ 
can be done in time polynomial in $\mo$ and exponential in $\mathcal N$ if $\mathcal N$ is an NBA. If $\mathcal N$ is good for MDPs, this can be done in time polynomial in $\mo$ and $\mathcal{N}$.
\end{restatable}

\noindent The corresponding matching lower bounds are inherited from the MDP case~\cite{Courco95}.

In the second step we perform quantitative analysis of SIMDPs. Similar to MDPs, this reduces to a reachability game but with an added adversarial nature.

\subsection{Interval Value Iteration}
Our first quantitative procedure is an interval value iteration algorithm (\Cref{alg:value-iteration-new}), in the spirit of the interval iteration method of~\cite{Haddad18}. In the first step, we identify two regions: $W^+$, the almost-sure winning region, and $W^-$, the sure losing region. The algorithm then initializes, for each state, a lower and an upper value. For states in $W^+$, both values are set to $1$; for states in $W^-$, both are set to $0$; and for all other states, the lower and upper values are initialized to $0$ and $1$, respectively.

In each iteration, the values of states in $W^+$ and $W^-$ remain fixed. For every other state, the values are updated as follows. First, for each $(s,a)$, nature chooses transition probabilities $\theta(s,a,s')$ adversarially within the prescribed intervals using the local best-response algorithm of~\cite{nilim-ghaoui}, applied separately to the current lower and upper values. This minimizes
$\sum_{s' \in S} \theta(s,a,s') V(s')$,
where $V$ is instantiated by the current lower or upper value function (\Cref{alg:nature-best-response}). The lower and upper values are then updated in the usual way by taking the weighted average of successor values. For the upper bound, we additionally multiply by a contraction factor $\alpha < 1$.

The algorithm terminates once the upper value is no greater than the lower value at every state. Standard arguments show that \Cref{alg:value-iteration-new} terminates after finitely many iterations. The contraction factor is essential here: without contraction, the upper value at $v_1$ remains $1$ in every iteration, while the lower value quickly drops to $0.1$, so that the gap eventually stabilizes. With contraction, however, the upper value eventually falls below $0.1$.

% In each iteration, the values of states in $W^+$ and $W^-$ remain unchanged. For all other states, the values are updated as follows. First, for each $(s,a)$ the transition probabilities $\theta(s,a,s')$ of nature are chosen adversarially within their respective intervals using the best local nature Algorithm~\cite{nilim-ghaoui} for lower and upper values respectively.
% This minimizes $\sum_{s' \in S} \theta(s,a,s') V(s')$ where $V(s')$ are lower or upper values respectively (\Cref{alg:nature-best-response})
% After this step, the lower(resp.\ upper) value is updated in the conventional way with the weighted average of the previous lower (resp.\ upper) values. 
% The upper value is also factored by a contraction factor $\alpha < 1$.

% The algorithm terminates when the upper value is no more than the lower value at every state.
% Using standard techniques it can be shown that \Cref{alg:value-iteration-new} terminates after a finite number of steps. 
% Without contraction, the upper value of $v_1$ always stays $1$ in every iteration while the lower value quickly falls to $0.1$.
% The difference of upper and lower values remains fixed after some step. But with contraction, the upper value eventually falls below $0.1$.
% This algorithm also provides basic guarantees: by taking $\alpha$ close enough to $1$, the final value is close enough to the actual value.
The algorithm also provides a basic approximation guarantee: by choosing $\alpha$ sufficiently close to $1$, the final value can be made arbitrarily close to the true value.
\begin{restatable}{theorem}{viTermination}
\label{thm:termination}
For an IMDP $\mo$, \Cref{alg:value-iteration-new} always terminates with a value $v \leq \textsf{Pr}_{\mo}$. Moreover, for every $\varepsilon > 0$, there exists $\alpha < 1$ such that
$v \geq (1-\varepsilon)\textsf{Pr}_{\mo}$.
\end{restatable}

\begin{algorithm}[t]
\caption{Value Iteration for IMDPs with A.S. and Losing Regions}
\label{alg:value-iteration-new}
\LinesNumbered
\KwIn{Regions $W^+$, $W^-$; discount factor $\alpha<1$}
\KwOut{Approximate reachability values $V^-$}

\ForEach{$s$}{
    $(V^-(s),V^+(s)) \gets
    \begin{cases}
        (1,1) & \text{if } s\in W^+,\\
        (0,0) & \text{if } s\in W^-,\\
        (0,1) & \text{otherwise.}
    \end{cases}$
}

\While{$\exists s \text{ such that } V^-(s)<V^+(s)$}{
    \ForEach{$s\notin W^+\cup W^-$}{
        \ForEach{action $a$ enabled at $s$}{
            compute $(\theta^-_{s,a,s'},\theta^+_{s,a,s'})$ via
            Algorithm~\ref{alg:nature-best-response}\;
        }
        $V^-(s)\gets \max_a \sum_{s'} \theta^-_{s,a,s'}V^-(s')$\;
        $V^+(s)\gets \max_a \alpha \sum_{s'} \theta^+_{s,a,s'}V^+(s')$\;
    }
}
\Return{$V^-$}\;
\end{algorithm}

\begin{algorithm}[t]
\caption{Best Local Nature Subroutine}
\label{alg:nature-best-response}
\KwIn{Values $V(S)_{s \in S}$}
\KwOut{Best local nature $\theta$ minimizing $\sum_{s' \in S} \theta(s,a,s') V(s')$ for each $s,a$}
\ForEach{$(s,a,s')$}{
        $\theta(s,a,s') \gets \underline{P}(s,a,s')$\;
    }
Sort states w.r.t. $V(s)$ into $s'_0,\dots,s'_{n-1}$ \; 
\ForEach{$(s,a)$}{
	$i \gets 0$\; 
	\While{$\sum_{s' \in S} \theta(s,a,s')$ < 1}{
$\begin{aligned}[t]
\theta(s,a,s'_i)
&\gets \theta(s,a,s'_i) + \\
&\quad \min\Bigl(
    \overline{P}(s,a,s'_i)-\theta(s,a,s'_i),\\
&\qquad 1-\sum_{s'\in S}\theta(s,a,s')
\Bigr)
\end{aligned}$
		$i \gets i+1$\; 
	}
}
\Return{$\theta$}\;
\end{algorithm}

\subsection{Nature Strategy Improvement Algorithm}
In our~\Cref{alg:policy-iteration-new}, we use strategy improvement for exact quantitative analysis. 
After solving the almost sure winning and sure losing regions,
we use strategy improvement, where each iteration (lines 4-7) consists of finding a new nature (lines 5,6) using the same subroutine \Cref{alg:nature-best-response} used for value iteration, followed by finding an optimal policy for this nature (line 7).

Existing strategy improvement \cite{nilim-ghaoui} for IMDP do this reversely: they iterate over policies and then look for the best nature response for these policies in every iteration.  %(\Cref{alg:policy-iteration-classic} in \Cref{app:algos}).
While this difference seems minor, computing the globally best response of nature is expensive: it entails choosing an optimal corner for nature among exponentially many vertices in the polytope formed by the $\underline{P}$ and $\overline{P}$ values, and this creates a linear program of exponential size (one variable for every corner of each polytope), whereas the linear program we generate in line 7 incurs no blow-up.

\begin{restatable}{theorem}{policyItNew}
\label{thm:policy-iteration-new}
\Cref{alg:policy-iteration-new} terminates with the optimal strategy and values. 
\end{restatable}
% \red{TO ADD PROOF IN APPENDIX}

In an ideal setting with exact rational arithmetic, the natural termination condition would be that the current policy remains optimal. Equivalently, one could stop when the value function does not change, for example when $\sum_{s \in S} V'(s) = \sum_{s \in S} V(s)$. (This is slightly more expensive as it may require solving the final linear program twice.)
However, in practical floating-point implementations, numerical precision can be insufficient to reliably determine whether nature should assign more probability mass to a successor $s$ or to $s'$ when their current values are extremely close.
Such marginal differences are purely due to numerical effects. They can lead to the solutions going in cycles with numerically indistinguishable values from the linear programs (e.g.\ one variable valuation alternating between 2 and 1.9999999999999999999999) that mask convergence, as it can lead to $\sum_{s \in S} V'(s) > \sum_{s \in S} V(s)$; we therefore use $\sum_{s \in S} V'(s) \geq \sum_{s \in S} V(s)$ as our stopping criterion. 
\begin{algorithm}[t]
\caption{Nature Strategy Improvement}
\label{alg:policy-iteration-new}
\LinesNumbered
\KwIn{$W^+,W^-$}
\KwOut{$\pi^*,V$}

Choose arbitrary $\theta'$\;
Compute best response $(\pi',V')$ to $\theta'$ via LP\;

\Repeat{$\sum_{s\in S}V'(s)\ge \sum_{s\in S}V(s)$}{
    $\theta,\pi,V \gets \theta',\pi',V'$\;
    \ForEach{$(s,a)$}{
        set $\theta'(s,a)$ to the best local nature for $V$\;
    }
    Compute best response $(\pi',V')$ to $\theta'$ via LP\;
}
    $\pi,V \gets \pi',V'$\;
\Return{$\pi,V$}\;
\end{algorithm}

\section{Experimental Evaluation}
\label{sec:experiments}
We implemented \Cref{alg:value-iteration-new,alg:policy-iteration-new}~\cite{experiment}. We evaluated the methods on two benchmarks from \textsc{DynAbs}~\cite{badings2023robust}, which constructs IMDP abstractions of stochastic dynamical systems.
We consider two case studies: (i) \emph{UAV motion control}, where a drone must be stabilized and guided safely to a target under non-Gaussian disturbances, and (ii) \emph{Shuttle control}, where a spacecraft must track a desired trajectory under uncertain noise. For the UAV benchmark, we varied the partition granularity to study scalability, obtaining models with \(787\) to \(2{,}430\) states and \(0.52\times 10^6\) to \(14.9\times 10^6\) transitions. The Shuttle benchmark contains \(3{,}200\) to \(7{,}200\) states and \(1\times 10^6\) to \(6\times 10^6\) transitions. Each IMDP was checked for stability to ensure that our quantitative analysis applies.

For both case studies, we considered a collection of \(\omega\)-regular specifications expressed in LTL and translated them into good-for-MDP B\"uchi automata using \textsc{Spot}~\cite{DuretLutz}. Figure~\ref{fig:combined_results} reports the runtime of SI, VI, and qualitative analysis for the UAV and Shuttle benchmarks as a function of the number of transitions, with each curve corresponding to a specification.
Overall, the results show that our framework scales to large IMDPs while supporting a broad range of temporal specifications. As expected, runtime increases with model size for all three procedures. Qualitative analysis remains relatively inexpensive, whereas quantitative analysis dominates the overall cost. SI is sometimes slower in practice because of its higher per-iteration cost, but its nature-improvement scheme makes its runtime comparable to VI.

% The results demonstrate that our approach scales to large IMDPs and supports a broad class of complex temporal specifications, with varying size and structure, including nondeterministic Büchi automata; Which shows practical applicability of our framework for realistic stochastic systems where both uncertainty and requirements must be handled simultaneously.
% Across all benchmarks, the runtime of SI, VI, and qualitative analysis grow with the size of the IMDP. Qualitative analysis remains relatively inexpensive, while the quantitative methods dominate the overall computation time.
% Although SI can be slower in practice, our iteration scheme over nature makes its performance comparable to VI.

The experiments also highlight the influence of the specification on computational cost. Simple reachability formulas, such as \(\mathbf{F}\,\mathit{reach}\), are handled efficiently, while formulas with nested temporal operators or combined safety and liveness requirements lead to larger automata and higher runtimes. In particular, specifications such as \(\mathbf{FG}(\mathit{reach}) \land \mathbf{G}\neg(\mathit{deadlock})\) are consistently among the most expensive.

% A key observation is the impact of LTL specifications on computational cost. Simple reachability properties (e.g., $\mathbf{F}\textit{(reach)}$) result in relatively fast convergence. 
% In contrast, specifications involving nested temporal operators and safety constraints induce more complex automata and higher computational cost.
% In particular, formulas that combine safety and liveness constraints 
% (such as $\mathbf{FG}(\textit{reach}) \land \mathbf{G}\neg \textit{(deadlock)}$) consistently exhibit the highest runtimes.

Supporting nondeterministic B\"uchi automata is important in practice, since determinization may incur exponential blow-up for many LTL formulas. By avoiding determinization, our framework remains expressive while retaining scalability, and can therefore handle richer specifications that would otherwise be impractical.

% Supporting nondeterminism is crucial in practice, as many LTL specifications cannot be translated into deterministic automata without an exponential blow-up. By avoiding determinization, our approach remains both expressive and scalable, enabling the analysis of richer specifications that would otherwise be infeasible. 
% Nondeterministic transitions allow the automaton to represent multiple possible evolutions of the specification simultaneously, which increases the branching in the product. While this can lead to larger state spaces and higher computational cost, it is essential for correctly capturing the semantics of general LTL specifications.

All experiments were run on the Barkla HPC cluster using AMD EPYC 9334 processors, with 8 CPUs and 64\,GB RAM per job under Python~3.12.

\begin{figure}[h]
        \centering
          \caption{Average initial-state values for VI lower bounds, VI upper bounds, and SI across iterations for an IMDP with 1600 states under $\mathbf{GF}(\textit{reach})$.}
        \label{fig:val-iter-uperlower-init}
        \includegraphics[width=\columnwidth,trim={0 0.3cm 0 2cm},clip]{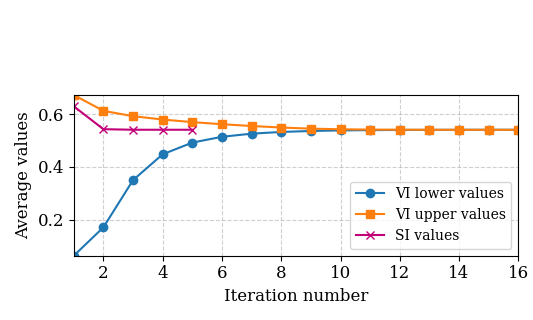}
    \end{figure}
    
\subsection{VI lower/upper bound values vs SI upper bound} 
To compare VI and SI more directly, we examine their convergence behavior. SI typically converges in fewer iterations, although each iteration is more expensive, leading to similar overall runtimes. Figure~\ref{fig:val-iter-uperlower-init} shows, for an IMDP with \(1{,}600\) states under the specification \(\mathbf{GF}(\mathit{reach})\), the average over the initial states of the VI lower bound, the VI upper bound, and the SI value at each iteration. The VI lower bound ($V^-$) increases monotonically, while the VI upper bound ($V^+$) decreases until the two meet, at which point \Cref{alg:value-iteration-new} terminates. The SI curve stabilizes rapidly, reaching its fixed point after about five iterations. The limit of the VI lower and upper bounds matches the SI solution. Thus, VI approaches the same value from below and above, whereas SI computes the exact value for each state.

% To compare methods, we can see SI and VI scaling trends. SI typically converges in fewer iterations but incurs a higher per-iteration cost, resulting in comparable overall runtimes. This confirms that both approaches are viable in practice, with trade-offs depending on the structure of the model and specification. Fig.~\ref{fig:val-iter-uperlower-init} shows the averaged value of VI lower bounds, VI upper bounds, and SI over the set of initial states across each iteration for an IMDP with $1{,}600$ states under the specification $\mathbf{GF}(\textit{reach})$. The VI lower bound increases and the VI upper bound decreases until they cross; Alg.~\ref{alg:value-iteration-new} stops at that point. The magenta curve shows the value obtained by SI at each of its iterations. SI converges fast: after 5 iterations, its value is already stable. The limit of VI’s lower/upper bounds coincides with the SI fixed point. VI eventually hits the same value, but SI gives the exact value for each state.

\begin{figure*}[h!]
    \centering

    % ----------- ROW 1: UAV -----------

    \begin{subfigure}[b]{0.32\textwidth}
        \centering
        \includegraphics[width=\linewidth]{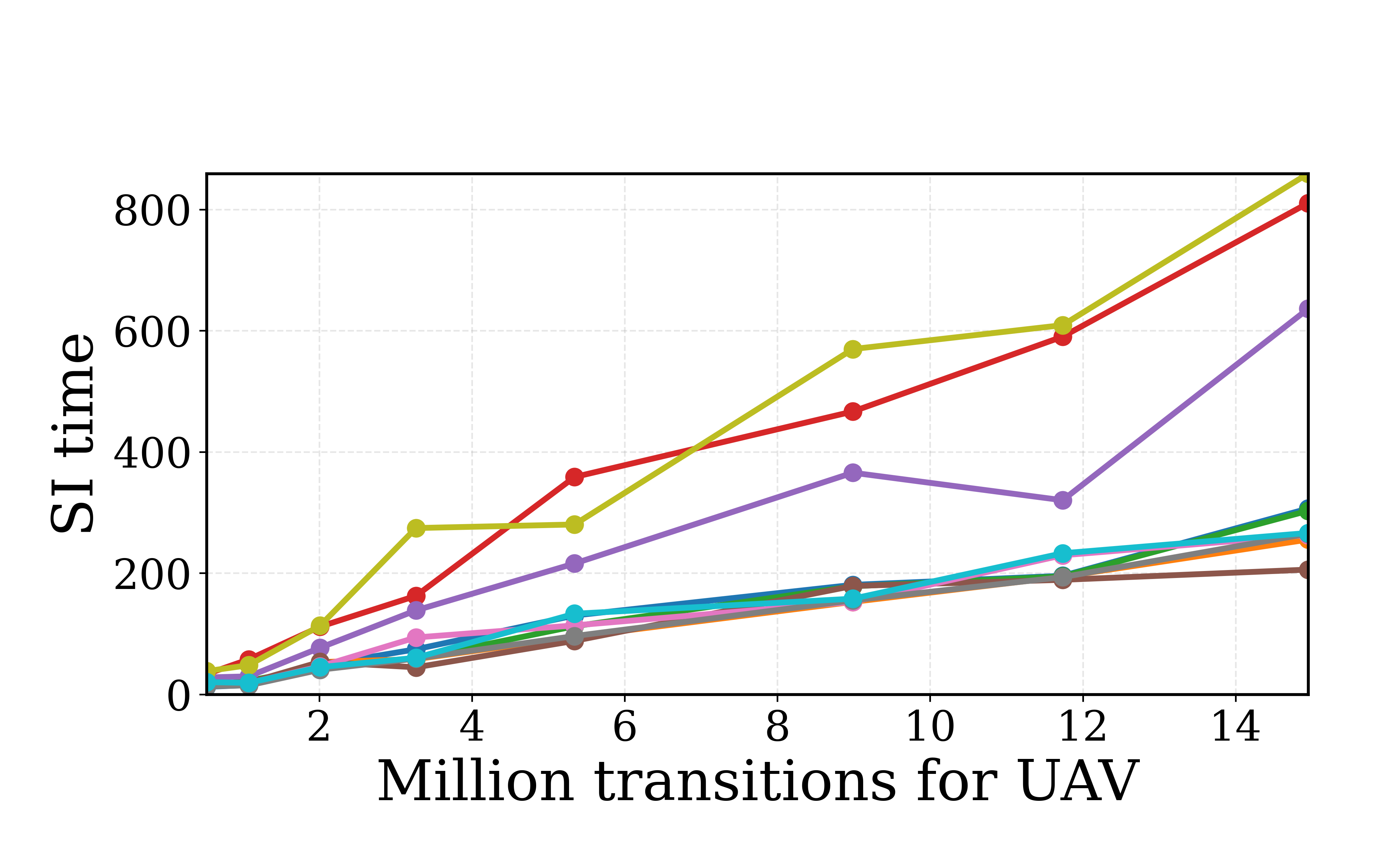}
    \end{subfigure}
    \hfill
    \begin{subfigure}[b]{0.32\textwidth}
        \centering
        \includegraphics[width=\linewidth]{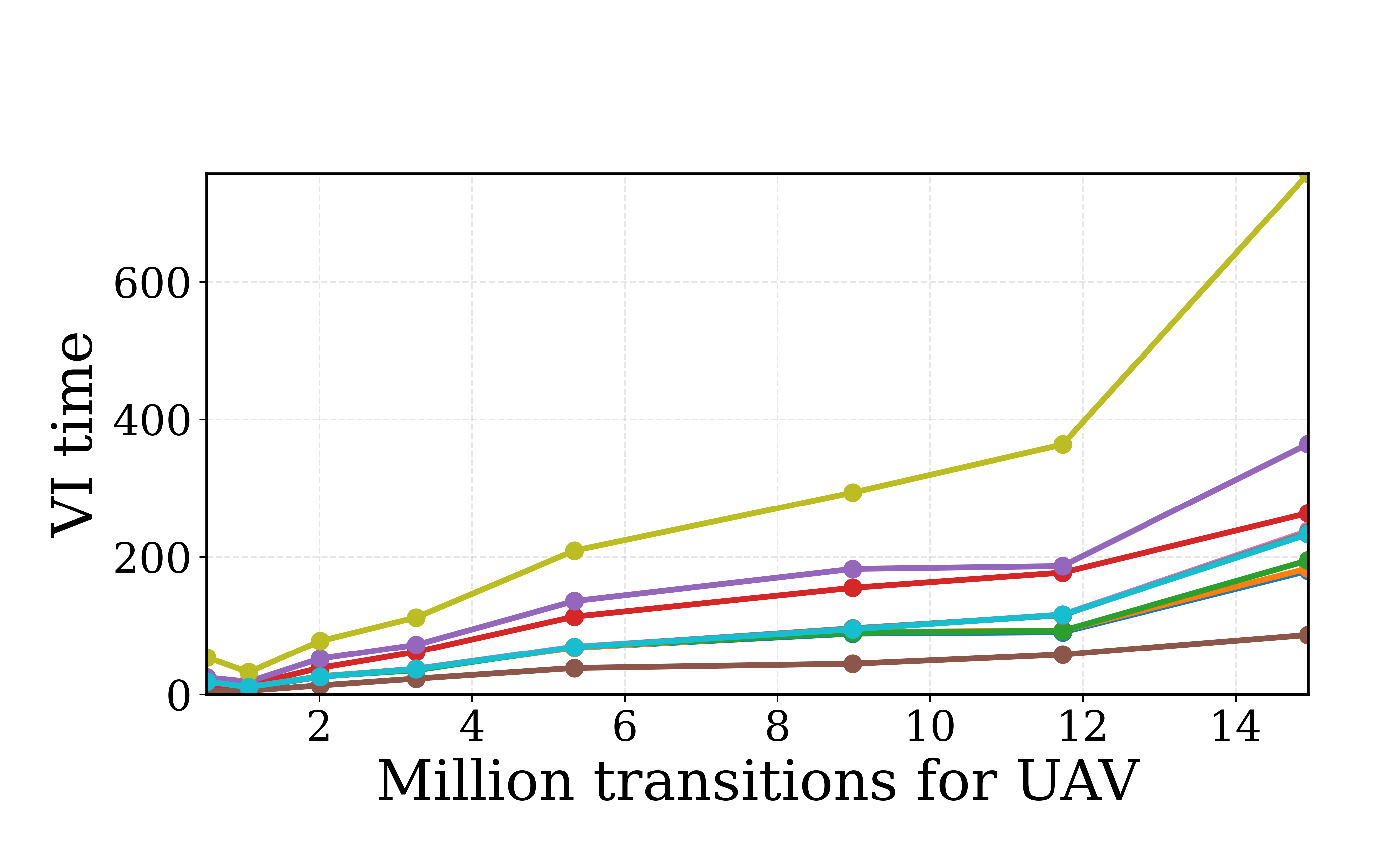}
    \end{subfigure}
    \hfill
    \begin{subfigure}[b]{0.32\textwidth}
        \centering
        \includegraphics[width=\linewidth]{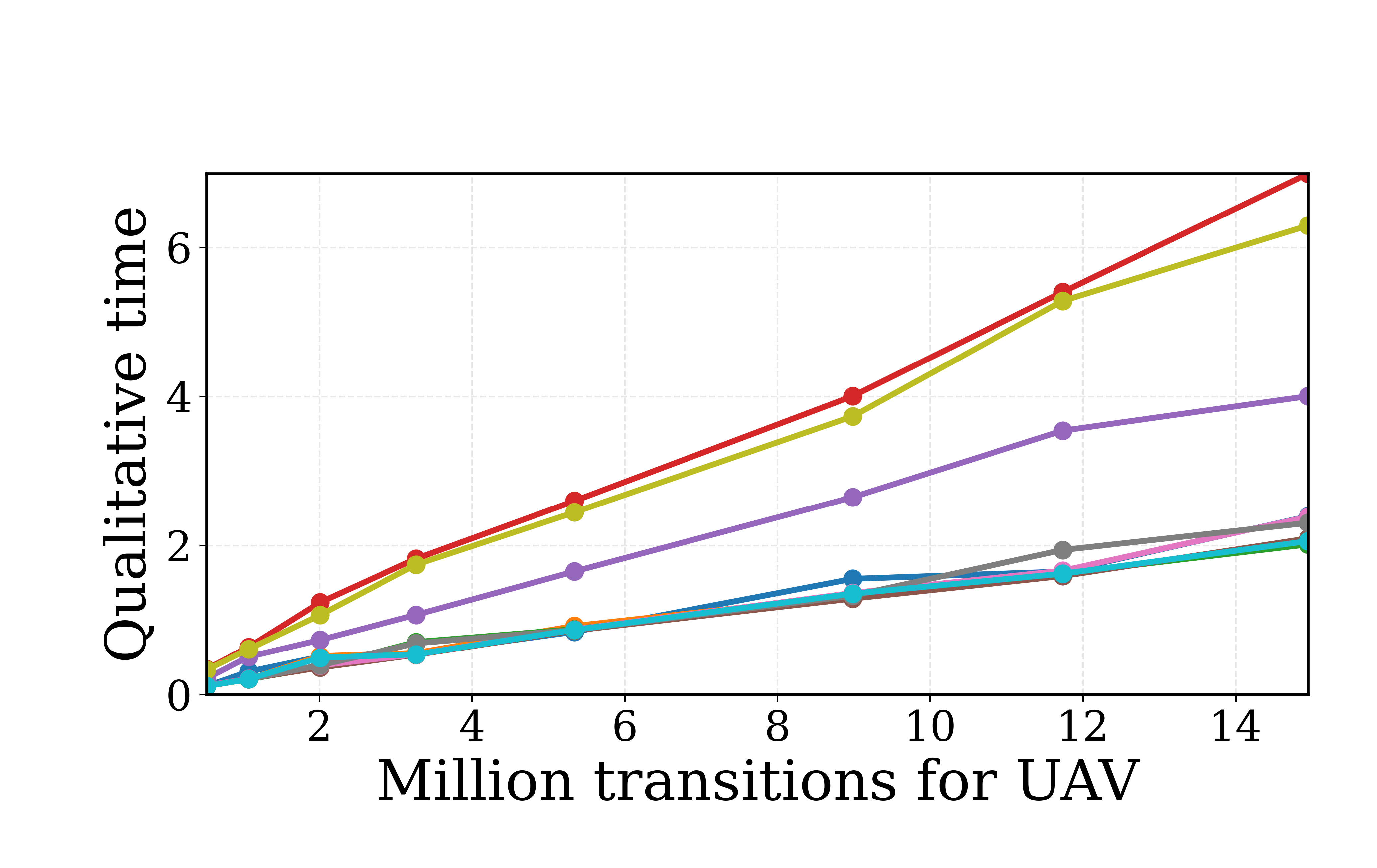}
    \end{subfigure}
    % \vspace{0.5em} 

    % ----------- ROW 2: SHUTTLE -----------

    \begin{subfigure}[b]{0.32\textwidth}
        \centering
        \includegraphics[width=\linewidth]{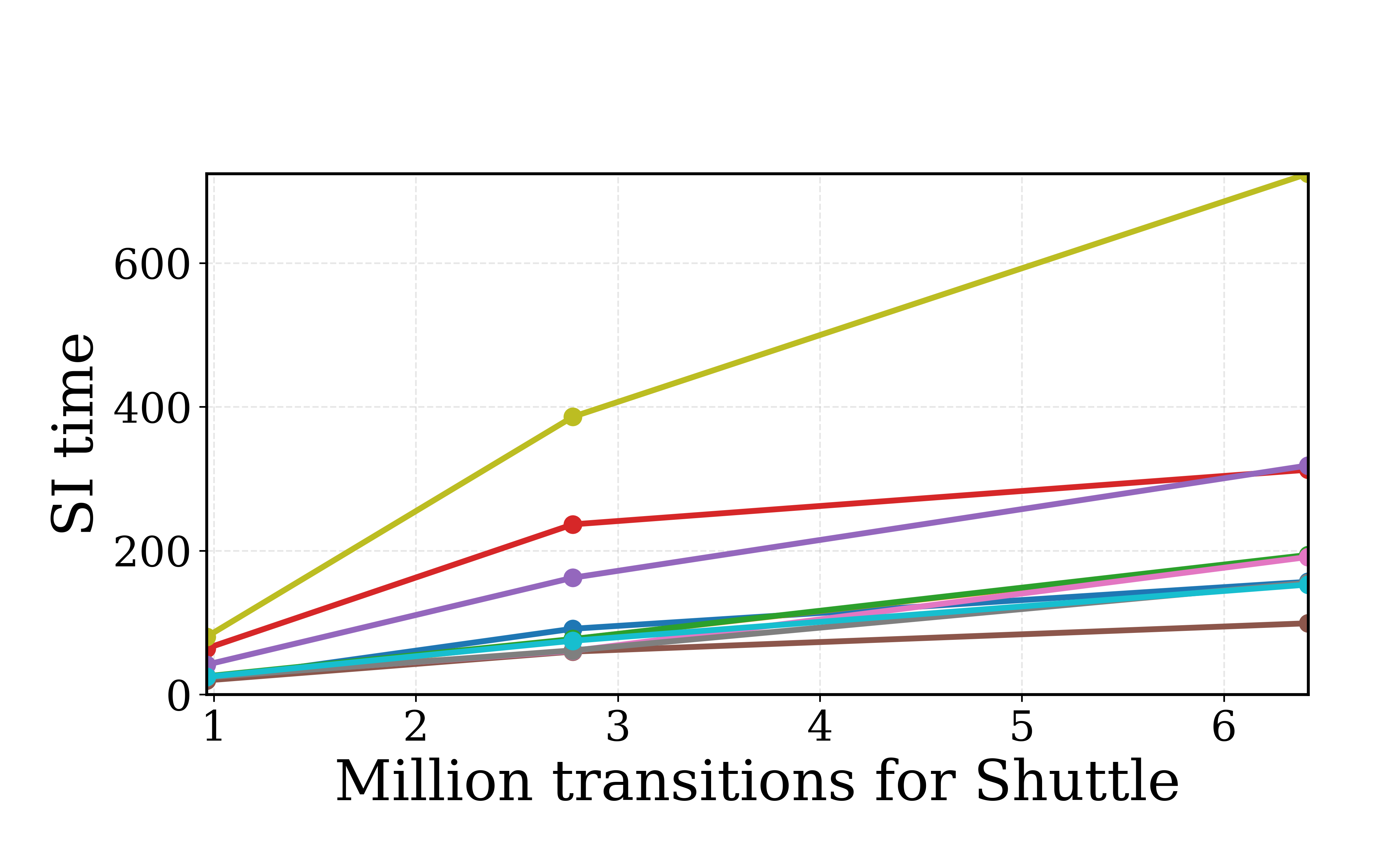}
    \end{subfigure}
    \hfill
    \begin{subfigure}[b]{0.32\textwidth}
        \centering
        \includegraphics[width=\linewidth]{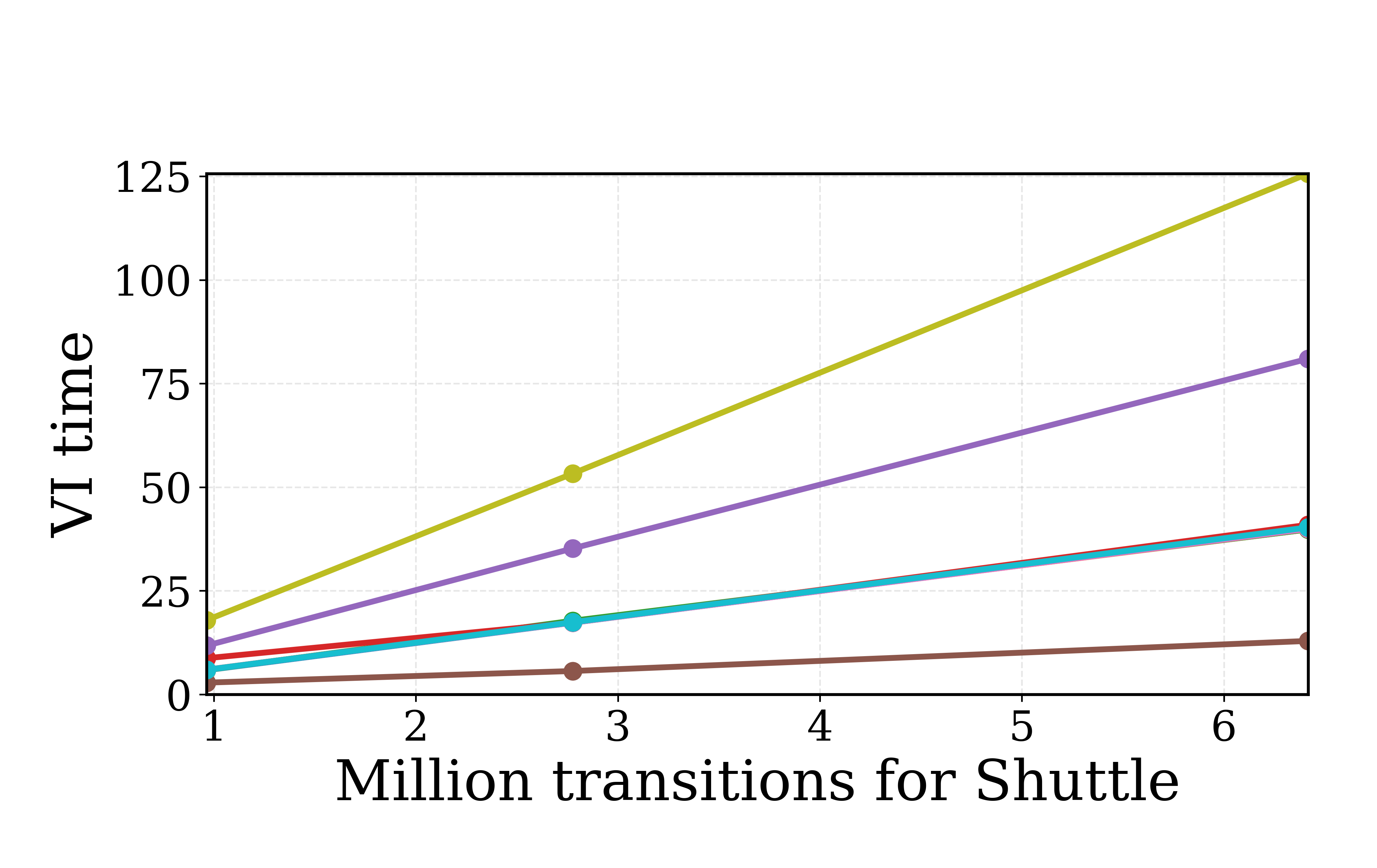}
    \end{subfigure}
    \hfill
    \begin{subfigure}[b]{0.32\textwidth}
        \centering
        \includegraphics[width=\linewidth]{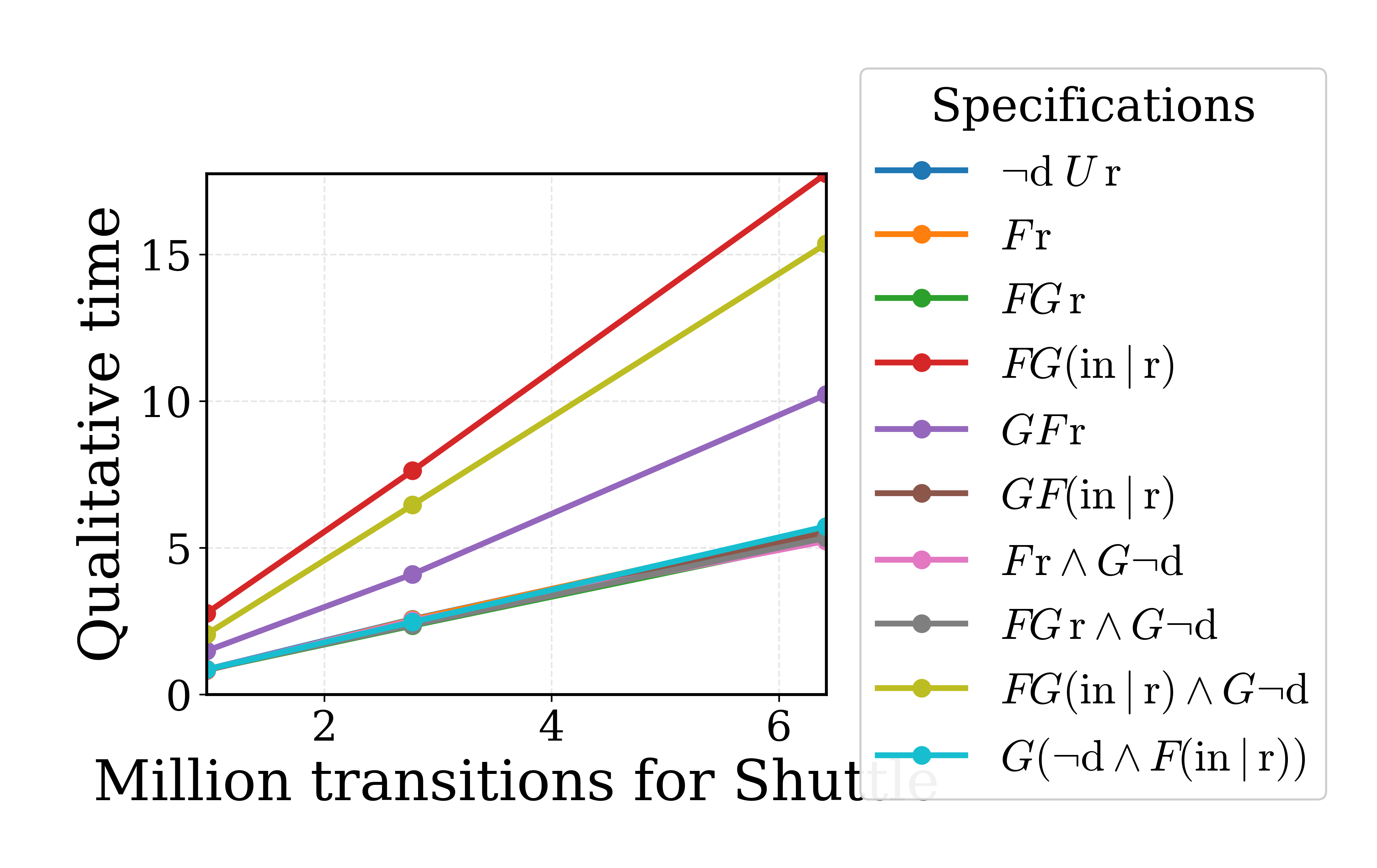}
    \end{subfigure}
    \caption{
Performance comparison for the UAV (top row) and Shuttle (bottom row) benchmarks under the selected specifications. The columns correspond to Strategy Improvement (left), Value Iteration (middle), and Qualitative Analysis (right).  Times are reported in seconds. The legend is shown only once; the same color-to-specification mapping holds for all plots.
    }
    \label{fig:combined_results}
\end{figure*}

\subsection{GR(1) Specifications}
\label{subsec:GR1}
Besides full LTL, we also evaluate our framework on the GR(1) fragment.
%, which naturally gives rise to nondeterminism.
%% Sven: commented that out -- GR(1) became popular because it has a relatively efficient determinisation.
We consider the following GR(1) specifications over the atomic propositions \textit{init}, \textit{reach}, and \textit{deadlock}:
\[
\begin{aligned}
\Phi_1:\quad
& \mathbf{GF}(init)\land \mathbf{GF}(\neg deadlock)
   \rightarrow \mathbf{GF}(reach), \\
\Phi_2:\quad
& \mathbf{GF}(\neg deadlock)
   \rightarrow \mathbf{GF}(reach)\land \mathbf{GF}(init), \\
\Phi_3:\quad
& \mathbf{GF}(\neg deadlock)
   \rightarrow \mathbf{GF}(reach \land \neg init), \\
\Phi_4:\quad
& \mathbf{GF}(init \land \neg deadlock)
   \rightarrow \mathbf{GF}(reach \land \neg deadlock).
\end{aligned}
\]

Here, $\Phi_1$ requires recurring success under recurring initialization and infinitely many non-deadlock states. $\Phi_2$ requires repeated recovery and success, namely infinitely many visits to both \textit{init} and \textit{reach}. $\Phi_3$ strengthens this by requiring recurring success outside the initial region. $\Phi_4$ is deadlock-sensitive: whenever the system can repeatedly return to a live initial condition, it must also repeatedly achieve a live success condition.
Figure~\ref{fig:gr1_uav_results} reports the runtime of SI, VI, and qualitative analysis on the UAV benchmarks. The results show that our framework handles GR(1)-induced nondeterminism in large IMDPs and supports richer reactive mission specifications.

\begin{figure*}[h]
    \centering    
    \begin{subfigure}[b]{0.32\textwidth}
        \centering
        \includegraphics[width=\linewidth]{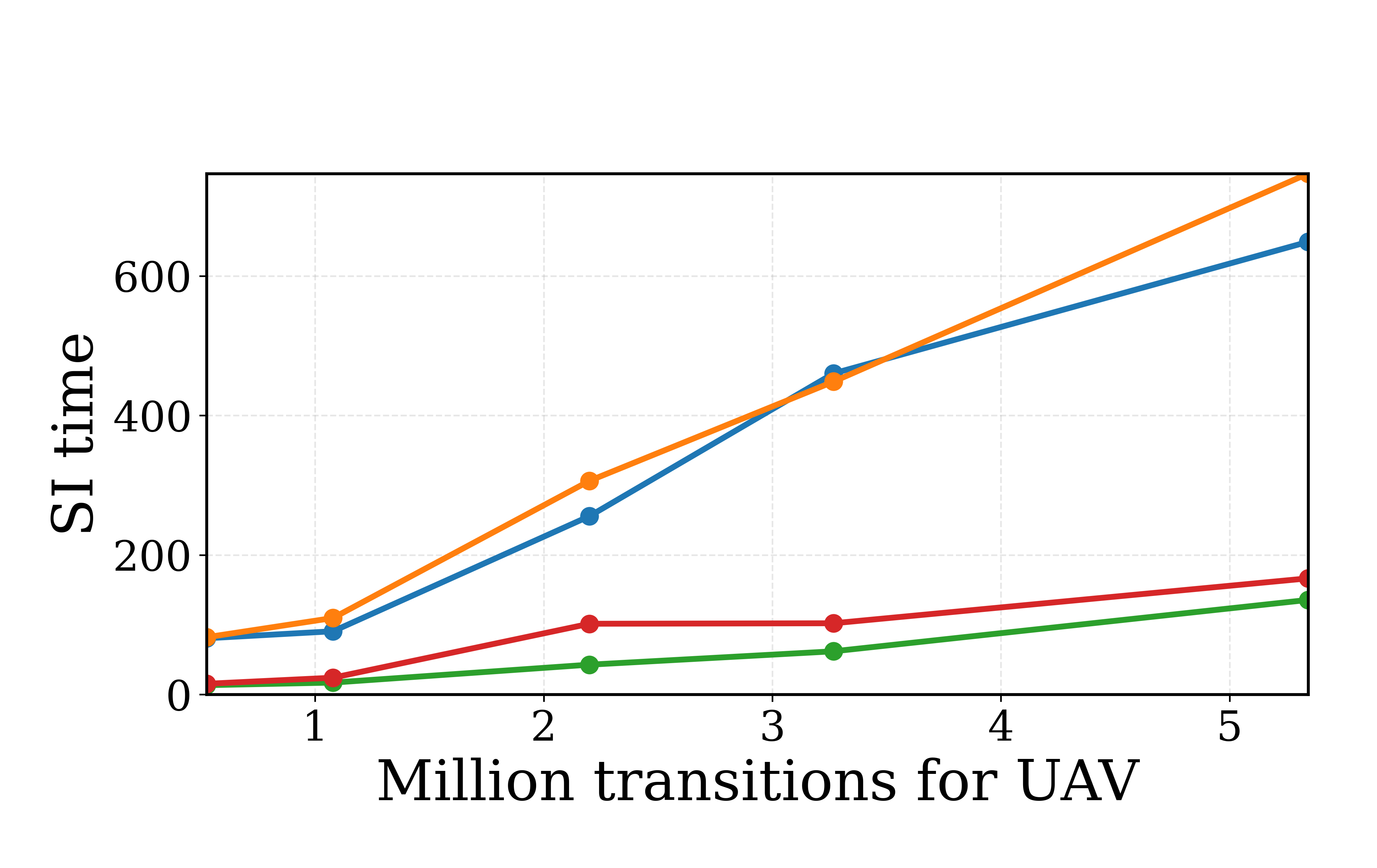}
    \end{subfigure}
    \hfill
    \begin{subfigure}[b]{0.32\textwidth}
        \centering
        \includegraphics[width=\linewidth]{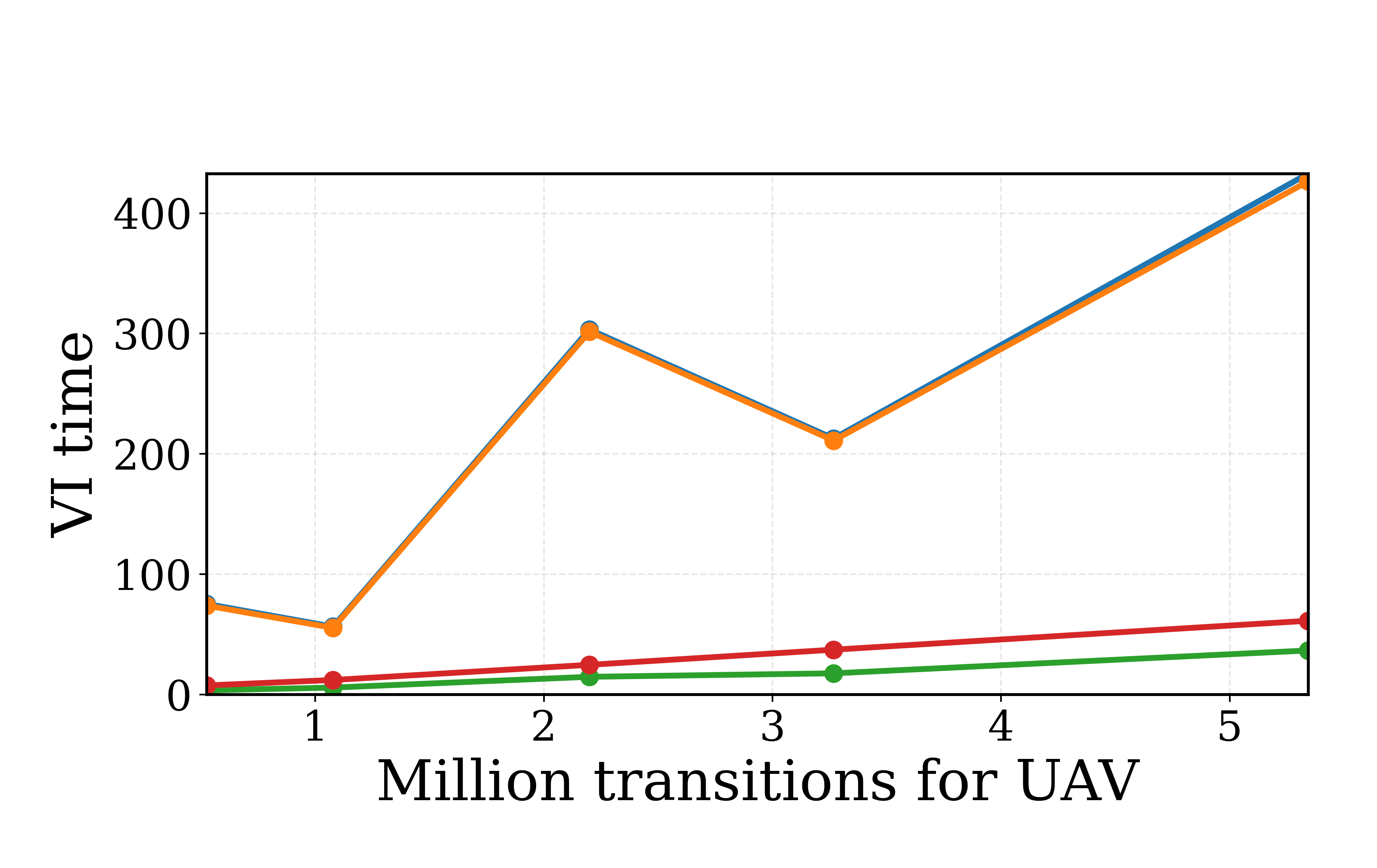}
    \end{subfigure}
    \hfill
        \begin{subfigure}[b]{0.32\textwidth}
        \centering
        \includegraphics[width=\linewidth]{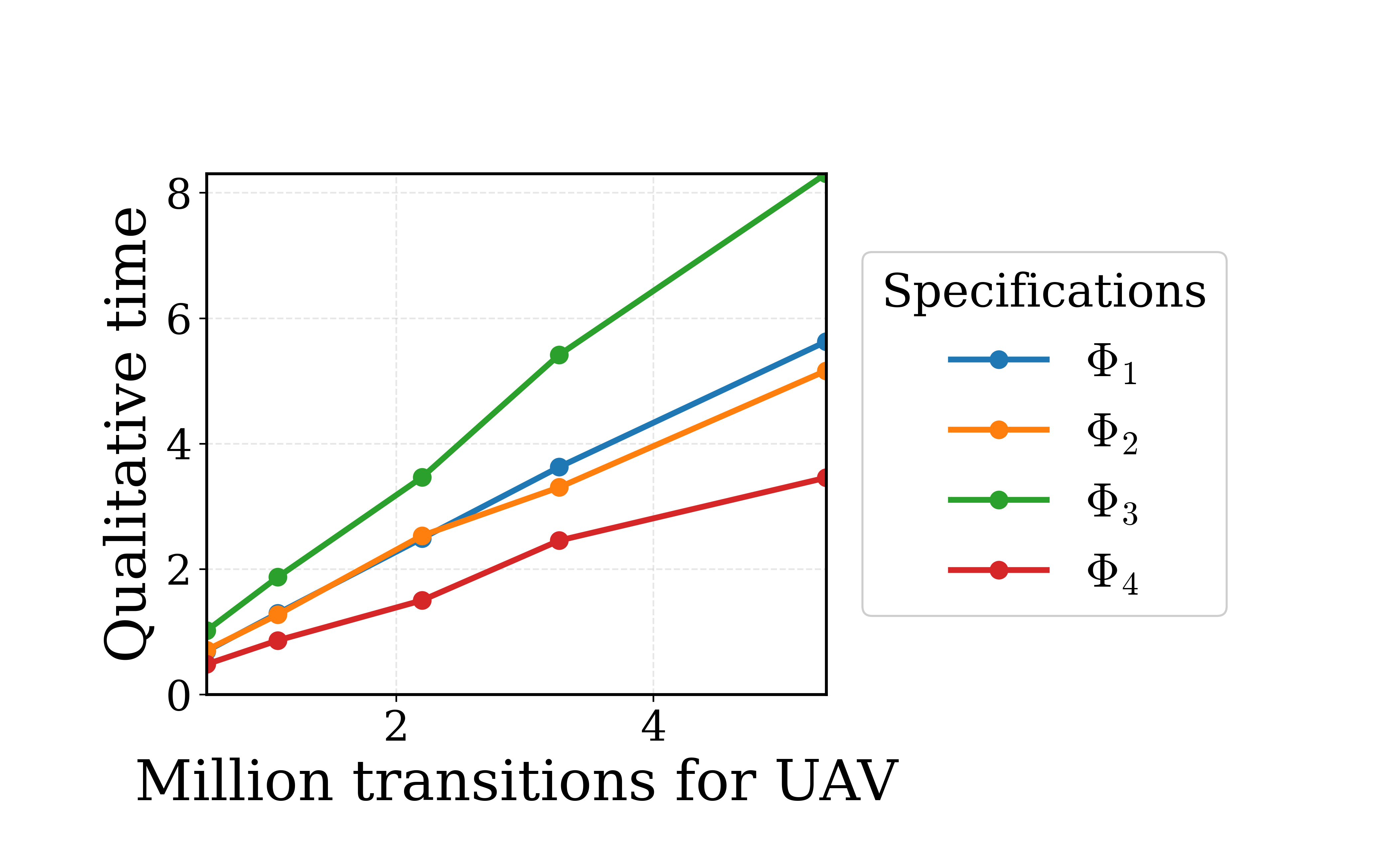}
    \end{subfigure}
    \caption{
Performance comparison for the UAV benchmarks under the GR(1) specifications as a function of the number of transitions. Each plot corresponds to Strategy Improvement (left), Value Iteration (middle), and Qualitative Analysis (right). Times are reported in seconds. The legend is shown only once; the same color-to-specification mapping holds for all plots.
    }
    \label{fig:gr1_uav_results}
\end{figure*}

\section{Conclusion}
We introduced a framework for verifying and synthesizing $\omega$-regular objectives over interval MDPs, motivated by the broader goal of supporting model-based reinforcement learning. Our analysis focused on \emph{stable} IMDPs, namely those whose concretizations share the same support, and showed that they admit a substantially simpler theory: qualitative analysis coincides with that of ordinary MDPs, while quantitative analysis reduces to reachability.

The significance of stability is not merely theoretical. IMDPs derived from model-based RL naturally satisfy this property, since every observed transition appears with positive probability in the learned dynamics. As a result, the abstraction preserves the connectivity structure needed for automata-theoretic verification. By making the relationship among IMDPs, MDPs, and stochastic games explicit, we also clarified when good-for-MDP automata suffice and when the stronger good-for-games condition becomes necessary for robust satisfaction.

On the algorithmic side, we developed and compared value-iteration- and strategy-improvement-based synthesis procedures. Our empirical evaluation on representative benchmarks shows that both approaches scale well and support end-to-end integration of formal specifications with uncertainty-aware decision making for stable IMDPs having millions of transitions under LTL objectives.

\bibliographystyle{IEEEtran}
\bibliography{papers}

% \end{document}doesn't compile otherwise for teh moment -- do remomve that enddocument when you need it :)
\newpage 
% \section{Appendix}

\appendix
%\textcolor{red}{Appendix should appear before Acks.}
%\section{Abstraction algorithms}
% \section{Proofs(if any)}

\subsection{Abstraction Algorithms}
Here we describe the abstraction algorithms. A schema can be found in \Cref{fig:imdp-abstraction}. 
We first recall the definition of A discrete-time stochastic control system.
\begin{definition}
A discrete-time stochastic control system (dtSCS) $\Sigma_{ss}$ is a tuple $(\mathcal S, U,w,f)$, where  
\begin{itemize}
    \item $\mathcal{S}\subseteq \mathbb{R}^n$ is the continuous state space;
    \item $U$ is the input space;
    \item $w$ is a sequence of i.i.d.\ random variables
    \[
    w:=\{w(k):\Omega \to V_w,\; k\in \mathbb{N}\},
    \]
    where $\Omega$ is the sample space and $V_w$ is the noise space;
    \item $f:\mathcal{S}\times U\times V_w \to \mathcal{S}$ is a measurable function describing the state evolution
    \begin{equation}
        x(k+1)=f(x(k),a(k),w(k)), \label{dynamic_evolu}
    \end{equation}
    where $k\in\mathbb{N}$, $x(k)\in\mathcal{S}$, $a(k)\in U$, and $w(k)\in V_w$.
\end{itemize}
\end{definition}
We denote by $\mathcal{U}_{\mathfrak a}$ the set of admissible input sequences
\[
\mathcal{U}_{\mathfrak a}
:=
\bigl\{\{a(k):\Omega\to U,\; k\in\mathbb{N}\}\bigr\},
\]
where, for each $k\in\mathbb{N}$, the input $a(k)$ is independent of $w(t)$ for all $t\ge k$.
We assume that the system~\eqref{dynamic_evolu} is unknown (the disturbance distribution is unknown) but data from sampled trajectories is available. 

\vspace{0.2em}\noindent\textbf{IMDP Abstraction.}
Depending on the assumptions on $\Sigma_{ss}$~\eqref{dynamic_evolu}, one can construct an IMDP abstraction using sampled data with probabilistic guarantees.

\begin{assumption}[Linear Systems]
\label{ass:linear}
Assume $f$ in~\eqref{dynamic_evolu} is linear of the form $f(x,a,w)=Ax+Ba+w$, where $A$ and $B$ are known but the distribution of $w$ is unknown.
Moreover, the pair of matrices $(A,B)$ is controllable.
\end{assumption}
Under Assumption~\ref{ass:linear}, an IMDP abstraction can be constructed using the data-driven procedure shown in Algorithm~\ref{alg:linear-imdp}.
\begin{theorem}[IMDP abstractions for linear systems]
Under Assumption~\ref{ass:linear}, and with the IMDP computed via Algorithm~\ref{alg:linear-imdp}, the number of sampled trajectories can be chosen such that every abstract transition probability is either exactly zero or belongs to an interval whose boundary does not include zero.
\end{theorem}
Intuitively, Algorithm~\ref{alg:linear-imdp} separates the existence of transitions—determined by the nominal linear dynamics—from the estimation of disturbance probabilities, which reduces to a Bernoulli success-probability estimation problem. With sufficiently many samples, the resulting confidence intervals avoid zero as a boundary point. 
See~\cite{badings2023robust} for full technical details.

\begin{assumption}[Nonlinear Systems]
\label{ass:nonlinear}
Assume $f$ in~\eqref{dynamic_evolu} is nonlinear and unknown, with a known upper bound on its Lipschitz constant. The distribution of $w$ is unknown.
\end{assumption}
Under Assumption~\ref{ass:nonlinear}, an IMDP abstraction can be constructed via the Lipschitz-based over-approximation procedure shown in Algorithm~\ref{alg:nonlinear-imdp}.
\begin{theorem}[IMDP abstractions for nonlinear systems]
Under Assumption~\ref{ass:nonlinear}, the IMDP constructed using  Algorithm~\ref{alg:nonlinear-imdp} may contain probability intervals whose boundary includes zero.
\end{theorem}

The presence of zero on interval boundaries stems from the conservatism induced by Lipschitz over-approximation (Step~2 of Algorithm~\ref{alg:nonlinear-imdp}), which may enlarge reachable sets enough to include transitions that are not observed in sampled data. 
See~\cite{nazeri2025data} for the complete development.

With the abstraction procedures in Algorithms~\ref{alg:linear-imdp} and~\ref{alg:nonlinear-imdp}, we obtain an IMDP representation of the underlying dynamical system together with interval-valued transition probabilities.
To synthesize a strategy that ensures satisfaction of temporal specifications on such an IMDP, the uncertainty in the transition probabilities must be treated adversarially.
This naturally leads to a game-theoretic formulation combined with an automaton encoding of the specification. 
Therefore, we introduce stochastic parity games and parity automata next.

% \begin{comment}
% \begin{figure*}[t]
% \centering
% \begin{minipage}[t]{0.48\textwidth}
\begin{figure*}[t]
\centering
\begin{tikzpicture}[>=latex,scale=0.95,every node/.style={font=\small}]

% -------------------------------------------------
% Styles
% -------------------------------------------------
\tikzstyle{region}=[draw, rounded corners=2pt, minimum width=1.2cm, minimum height=1.2cm]
\tikzstyle{absnode}=[draw, circle, minimum size=8mm, inner sep=0pt]
\tikzstyle{sample}=[circle, fill=black, inner sep=1.2pt]
\tikzstyle{arrowinterval}=[->, thick]
\tikzstyle{arrowplain}=[->, thin]

% -------------------------------------------------
% Left: continuous state space with samples
% -------------------------------------------------
\node[anchor=west,font=\bfseries] at (-0.2,4.3) {Continuous state space};
\draw[thick] (0,0) rectangle (4,4);

% Axes labels
\node at (2,-0.35) {$x_1$};
\node[rotate=90] at (-0.35,2) {$x_2$};

% Sample points
\foreach \x/\y in {
0.5/0.7, 0.8/1.4, 1.1/2.8, 1.5/0.9, 1.7/2.1,
2.2/1.0, 2.4/2.5, 2.8/3.0, 3.0/1.7, 3.3/2.4,
3.5/0.8, 2.9/0.6
}{
  \node[sample] at (\x,\y) {};
}

% A faint nonlinear trajectory cloud / arrows
\draw[blue!60, very thick, ->] (0.7,0.8) .. controls (1.3,1.2) and (1.7,1.6) .. (2.2,2.0);
\draw[blue!60, very thick, ->] (1.2,2.7) .. controls (1.8,2.9) and (2.2,3.0) .. (2.9,3.1);
\draw[blue!60, very thick, ->] (2.1,1.0) .. controls (2.5,1.3) and (2.8,1.5) .. (3.2,1.9);

% \node[align=center] at (2,-0.95) {Observed transitions / samples};

% -------------------------------------------------
% Middle: partitioned space
% -------------------------------------------------
\node[anchor=west,font=\bfseries] at (5.2,4.3) {Partition into regions};
\draw[thick] (5.5,0) rectangle (9.5,4);

% Grid
\foreach \x in {6.5,7.5,8.5}
  \draw[gray!70] (\x,0) -- (\x,4);
\foreach \y in {1,2,3}
  \draw[gray!70] (5.5,\y) -- (9.5,\y);

% Highlight a few cells
\fill[blue!12] (6.5,1) rectangle (7.5,2);
\fill[green!12] (7.5,2) rectangle (8.5,3);
\fill[red!12] (8.5,1) rectangle (9.5,2);

\node at (7,1.5) {$q_i$};
\node at (8,2.5) {$q_j$};
\node at (9,1.5) {$q_k$};

% Sample points mapped into cells
\foreach \x/\y in {
5.9/0.7, 6.2/1.4, 6.6/2.8, 7.0/0.9, 7.2/2.1,
7.7/1.0, 7.9/2.5, 8.3/3.0, 8.5/1.7, 8.8/2.4,
9.0/0.8, 8.4/0.6
}{
  \node[sample] at (\x,\y) {};
}

% \node[align=center] at (7.5,-0.95) {Finite abstraction by partitioning};

% Arrow from left to middle
\draw[arrowplain] (4.3,2) -- (5.2,2);
\node[align=center] at (4.75,2.35) {\scriptsize discretize};

% -------------------------------------------------
% Right: IMDP abstraction
% -------------------------------------------------
\node[anchor=west,font=\bfseries] at (10.7,4.3) {IMDP abstraction};

\node[absnode, fill=blue!12]  (qi) at (11.5,1.7) {$q_i$};
\node[absnode, fill=green!12] (qj) at (14.3,2.7) {$q_j$};
\node[absnode, fill=red!12]   (qk) at (14.3,0.7) {$q_k$};
\node[absnode, fill=yellow!15](ql) at (17.0,1.5) {$q_\ell$};

% Interval-labeled transitions
\draw[arrowinterval] (qi) to[bend left=15]
  node[midway, above, sloped] {\scriptsize $a,\;[0.25,0.45]$} (qj);

\draw[arrowinterval] (qi) to[bend right=15]
  node[midway, below, sloped] {\scriptsize $a,\;[0.30,0.55]$} (qk);

\draw[arrowinterval] (qj) to[bend left=12]
  node[midway, above, sloped] {\scriptsize $b,\;[0.40,0.70]$} (ql);

\draw[arrowinterval] (qk) to[bend right=10]
  node[midway, below, sloped] {\scriptsize $b,\;[0.10,0.35]$} (ql);

% Self-loop
\draw[arrowinterval] (qi) edge[loop below] node[below] {\scriptsize $a,\;[0.10,0.25]$} (qi);

% Arrow from middle to right
\draw[arrowplain] (9.9,2) -- (10.6,2);
\node[align=center] at (10.25,2.35) {\scriptsize estimate\\[-1mm]\scriptsize intervals};
\end{tikzpicture}
\caption{Construction of an IMDP abstraction from sampled data: the continuous state space is partitioned into finitely many regions, and transition probabilities under each control input are bounded by intervals, yielding a finite model amenable to verification and policy synthesis under uncertainty. %\Sadegh{Modify the figure to make it attractive.}
}
\label{fig:imdp-abstraction}
\end{figure*}
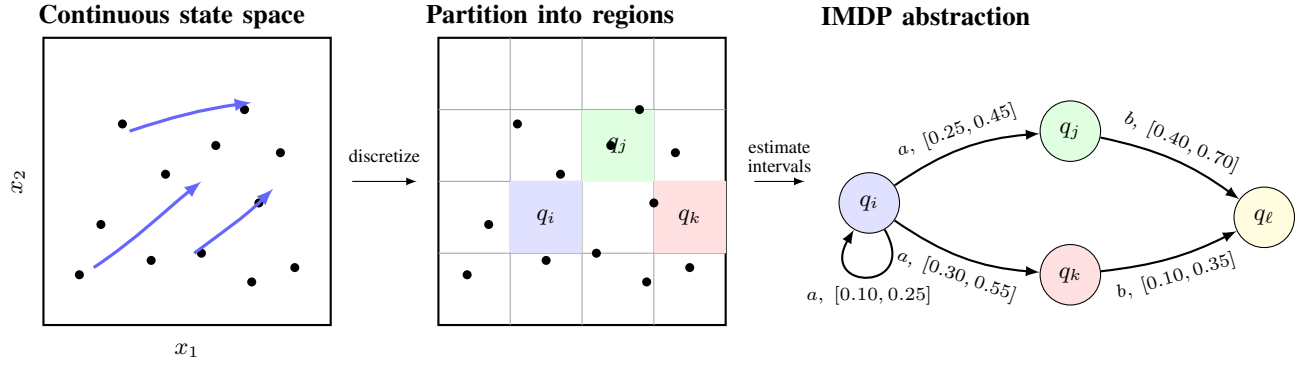

\begin{algorithm}
\caption{Abstraction of Linear Systems as IMDPs}
\label{alg:linear-imdp}
\LinesNumbered
\DontPrintSemicolon
\KwIn{Linear $f(x,a,w)=Ax+Ba+w$}
\KwOut{Abstract IMDP}

Construct partition $\set{S_1,\ldots,S_n}$ of the state space and select representative points $s_i \in S_i$.\;

\For{each $S_i$ and each $a\in U$}{
  Use $Ax+Ba$ to test whether $(Ax{+}Ba){\in} S_j$ for some $x{\in} S_i$.\;
  \If{so}{Declare abstract transition $S_i \xrightarrow{a} S_j$.\;}
}

Compute confidence intervals for $\Pr[w\in[\alpha,\beta]]$ from samples.\;

Combine Steps 2--3 to find abstract transition probabilities.\;

\end{algorithm}
% \end{minipage}
% \hfill
% \begin{minipage}[t]{0.51\textwidth}
\begin{algorithm}
\caption{Abstraction of Nonlinear Systems as IMDPs}
\label{alg:nonlinear-imdp}
\LinesNumbered
\DontPrintSemicolon
\KwIn{Nonlinear $f$ with known Lipschitz bound}
\KwOut{Abstract IMDP}
Construct partition $\{S_1,\ldots,S_n\}$ of the state space and select representative points $s_i \in S_i$.\;

\For{each sampled pair $(x,a)$}{
  Use Lipschitz bound to compute an over-approximation of the forward reachable set.\;
}

Using sampled trajectories, compute confidence intervals for transition probabilities
consistent with the over-approximated reachable sets.\;
Combine reachable-set information and statistical estimates to obtain probability
intervals between abstract states $S_i$ and $S_j$.\;
\end{algorithm}

% \end{minipage}

% \caption{Algorithms~\ref{alg:linear-imdp} and~\ref{alg:nonlinear-imdp} illustrate how IMDPs are constructed from sampled trajectories under different assumptions on the dynamics.}
% \vspace{-2em}
% \label{fig:imdp-abstraction}
% \end{figure*}
% \end{comment}

\subsection{Proof of Theorem~\ref{thm:qualitative-SIMDP}}
\qualitativeSIMDP*

\begin{proof}
We first observe that, for the qualitative analysis, only the support of the probability distribution matters, not the probabilities themselves, cf.~\cite{Chatte04,Hahn16}.
As a consequence, the choices of the environment in the game are irrelevant and we can instead work with any distribution from $\pd(x,a)$.
In particular, we can use the same one every time, and as they all have the same support, it does not matter which one we choose. In particular, we can choose arbitrarily and consistently for each state action pair $(x,a)$ for all states $(x,q,a)$ (i.e.\ independent of the automata state $q$ in the product).

This brings us back to the qualitative analysis of ordinary MDP. For the complexity and hardness for ordinary NBA see \cite{Courco95,Hahn15}, for good-for-MDP automata see \cite{Hahn20}.
%\qed
\end{proof}

\subsection{Proof of Theorem~\ref{thm:termination}}
%\viTermination*
\begin{proof}
After $n$ steps, the upper bound $V^+$  optimizes $\sum_{i=0}^n \alpha^i p_i'$, where $p_i'$ for $i<n$ is the chance of reaching the a.s.\ winning area in precisely $i$ steps, while $p_n'$ is the chance of not having reached the sure losing area after $n$ steps apply. The lower bound $V^-$ optimizes $\sum_{i=0}^n p_i$, where $p_i$ (including for $i=n$) is the chance of reaching the a.s.\ winning area in precisely $i$ steps.
So, the upper bound converges from above against the discounted reachability value, while the lower bound converges from below to the undiscounted reachability value.
As long as not all initial states are in the a.s. winning area or sure losing area, the limit value of the former is lower than the limit value of the latter, which provides termination of the algorithm.

Now, the value obtained is always strictly larger than the discounted reachability value for the given discount factor $\alpha$.
Given $\varepsilon$, in order to get suitable $\alpha$, we do the following :  pick an $n$ such that $\sum_{i=0}^n p_i \geq (1-0.5 \varepsilon)\textsf{Pr}_{\mo}$
holds for an optimal positional policy to realize $\textsf{Pr}_{\mo}$, and thus also for any policy to optimize $\sum_{i=0}^n p_i$, which is what $V^-$ gives after $n$ steps.

Then pick an $\alpha < 1$ s.t. $\alpha^n \ge (1-0.5\varepsilon)$.
Then $\sum_{i=0}^n \alpha^i p_i \ge \sum_{i=0}^n \alpha^n p_i \ge \alpha^n \cdot (1-0.5\varepsilon) \textsf{Pr}_{\mo} \ge (1-\varepsilon)\textsf{Pr}_{\mo}$, where the $p_i$ is the chance of reaching the a.s.\ winning region in $i$ steps. The estimation holds for all policies, and thus also for the supremum.
   % \qed
\end{proof}

\subsection{Proof of \Cref{thm:policy-iteration-new}}

\begin{proof}
The proof follows standard fixed-point characterization of stochastic reachability games.
Define the operator
\[
F(v)(s)=
\begin{cases}
1, & s\in W^+,\\
0, & s\in W^-,\\[1mm]
\displaystyle \max_{a\in A(s)} \min_{p\in \Delta(s,a)}
\sum_{t\in S} p(t)\,v(t), & \text{otherwise,}
\end{cases}
\]
Its least fixed point is the value vector $V^*$ of the game.

For a fixed positional nature strategy $\theta$, let
\[
T_\theta(v)(s)=
\begin{cases}
1, & s\in W^+,\\
0, & s\in W^-,\\[1mm]
\displaystyle \max_{a\in A(s)} \sum_{t\in S}\theta(s,a,t)\,v(t),
& \text{otherwise.}
\end{cases}
\]
Let $V^\theta$ be the value vector obtained by solving the MDP against $\theta$,
which is the least fixed point of $T_\theta$.

Let $(\theta_k,\pi_k,V_k)$ be the nature strategy in the $k$th iteration, $\pi_k$  a best
response to $\theta_k$ and $V_k=V^{\theta_k}$. For every state-action pair $(s,a)$,
\[
\sum_t \theta_{k+1}(s,a,t)\,V_k(t)
=
\min_{p\in \Delta(s,a)} \sum_t p(t)\,V_k(t).
\]
Hence
\[
T_{\theta_{k+1}}(V_k)=F(V_k)\leq T_{\theta_k}(V_k)=V_k.
\]
Since $V_{k+1}=V^{\theta_{k+1}}$ is the least fixed point of $T_{\theta_{k+1}}$
and $V_k$ is a post-fixed point, we obtain for all $s$,
\[
V_{k+1}(s)\leq V_k(s)
\]
Thus, the vectors $V_k$ are monotonically non-increasing.

If $V_{k+1}=V_k=:V$, then
\[
V=T_{\theta_{k+1}}(V)=F(V),
\]
So $V$ is a fixed point of $F$, and thus $V=V^*$.
Moreover $\pi_{k+1}$ is an optimal policy and $V$ is the optimal value vector.

For each $(s,a)$, the feasible set
$\Delta(s,a)$ is a polytope, and the linear minimization is
attained at a vertex. Therefore only finitely many positional nature strategies
can occur. Since each non-terminal iteration satisfies $V_{k+1}\leq V_k$ and
strictly decreases at least one component, no step can repeat before
termination. Hence the algorithm terminates after finitely many steps.

\end{proof}

\section{Sub-routines}

\begin{algorithm}
\caption{Classic Strategy Improvement}
\label{alg:policy-iteration-classic}
\KwIn{Regions $W^+$ (a.s. winning), $W^-$ (surely losing)}
\KwOut{Optimal strategy $\pi^*$ and value vector $V$}

$\pi \gets$ arbitrary policy\;
$\theta \gets$ arbitrary policy of nature\;
$V \gets$ values of states, $V(s) = \textsf{Pr}_{\mo}^{\pi,\theta}(s)$\;

\While{$\exists s \notin W^+ \cup W^-$ s.t.\ $\pi(s) \neq \arg \max_{a \in A(s)} \sum_{s' \in S} \theta(s,a,s') V(s')$}{
    \ForEach{state $s \notin W^+ \cup W^-$}{
        $\pi(s) \gets \arg \max_{a \in A(s)} \sum_{s' \in S} \theta(s,a,s') V(s')$\;
        $\theta \gets$ optimal nature against $\pi$\;
        $V \gets$ updated values under $\pi$ and $\theta$\;
    }
}

\Return{$\pi, V$}\;
\end{algorithm}

\subsection{Details on Experiments}

This appendix provides a detailed evaluation of the UAV motion-control benchmark introduced in \Cref{sec:experiments}. 
We vary the abstraction granularity in the continuous state space, producing SIMDPs with $787$--$2{,}430$ states and $0.52 \times 10^6$--$14.9 \times 10^6$ transitions.
All benchmarks are stable IMDPs, ensuring that the quantitative results reported here represent exact optimal probabilities.
We present performance trends for qualitative solving, value iteration (VI), and strategy improvement (SI), including convergence time, iteration counts, and value evolution across model sizes.
These additional plots complement the summary results in the main text and illustrate the scalability of our techniques on progressively refined abstractions of UAV dynamics.
\begin{itemize}
\item \textbf{Qualitative solving time.}
Figure~\ref{fig:Qual-Time-vs-criteria} reports the runtime of the qualitative analysis versus SIMDP size (left: states; right: transitions). Transition counts range from $0.52\!\times\!10^6$ (787 states) to $14.9\!\times\!10^6$ (2{,}430 states).

\item \textbf{VI runtime.}
Figure~\ref{fig:Val_iter_time_vs_criteria} shows VI convergence time versus states and transitions. As the abstraction granularity increases, transitions per state grow and runtime scales accordingly.

\item \textbf{VI iteration counts.}
Figure~\ref{fig:Val_iter-cnvrg-iter-vs-criteria} plots the number of iterations required for VI convergence. Both measures increase with model size, reflecting slower contraction in larger abstractions.

\item \textbf{VI lower/upper value bounds vs.\ SI upper bound.}
Figure~\ref{fig:app-val-iter-uperlower-init} shows the evolution of VI lower and upper bounds on a representative model with $1{,}600$ states. The bounds converge from opposite directions while SI reaches the exact fixed point in only a few iterations.

\item \textbf{SI runtime.}
Figure~\ref{fig:si-time-vs-criteria} reports SI convergence time versus states and transitions on the same benchmarks.

\item \textbf{SI vs.\ VI runtime ratio.}
Figure~\ref{fig:SI-over-VI-ratio} plots the ratio of SI runtime to VI runtime for each model. Ratios below $1$ indicate that SI consistently outperforms VI across all benchmarks.

\end{itemize}

\begin{figure}[h!]
\centering
    \includegraphics[scale=0.45]{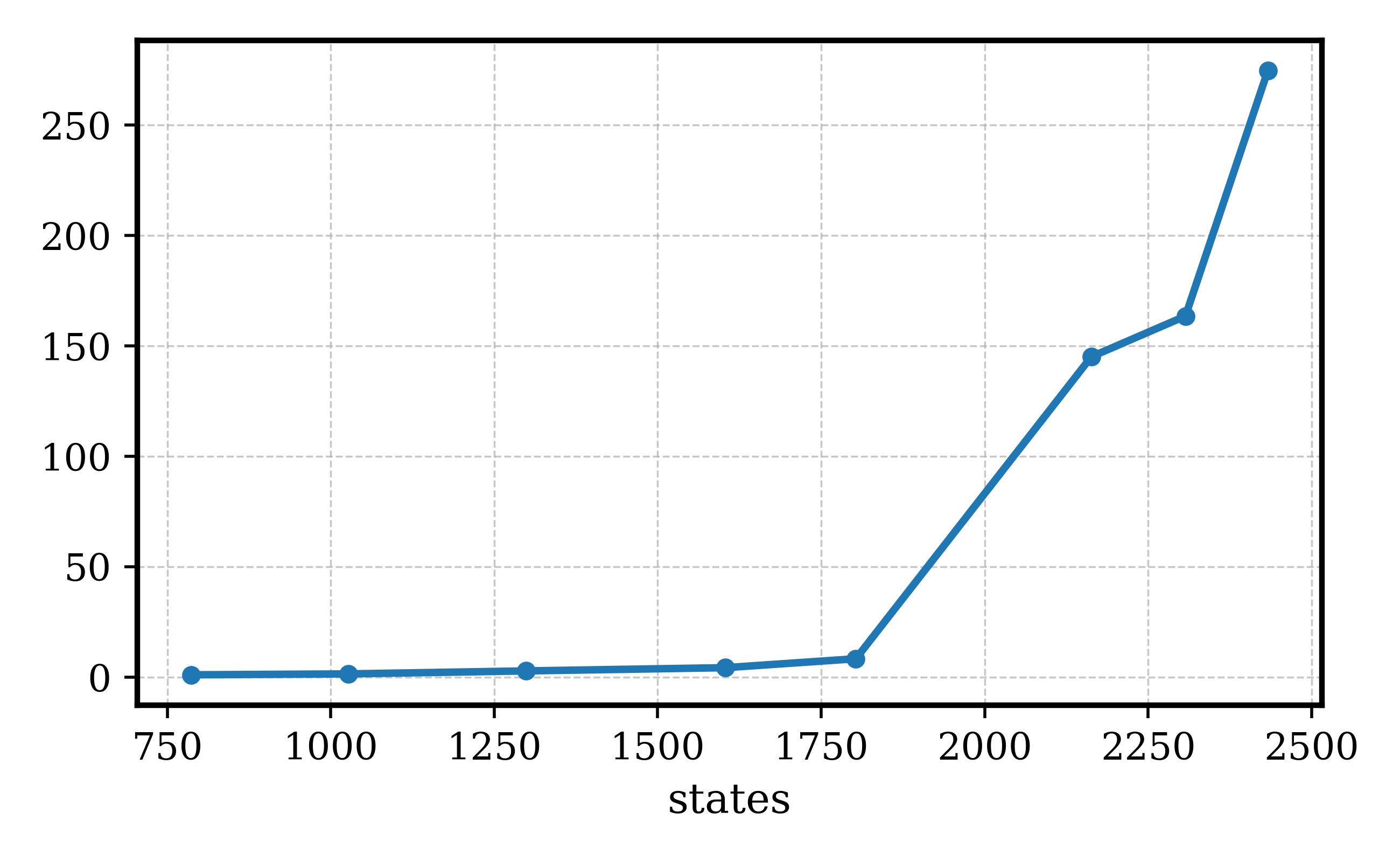}
    \includegraphics[scale=0.45]{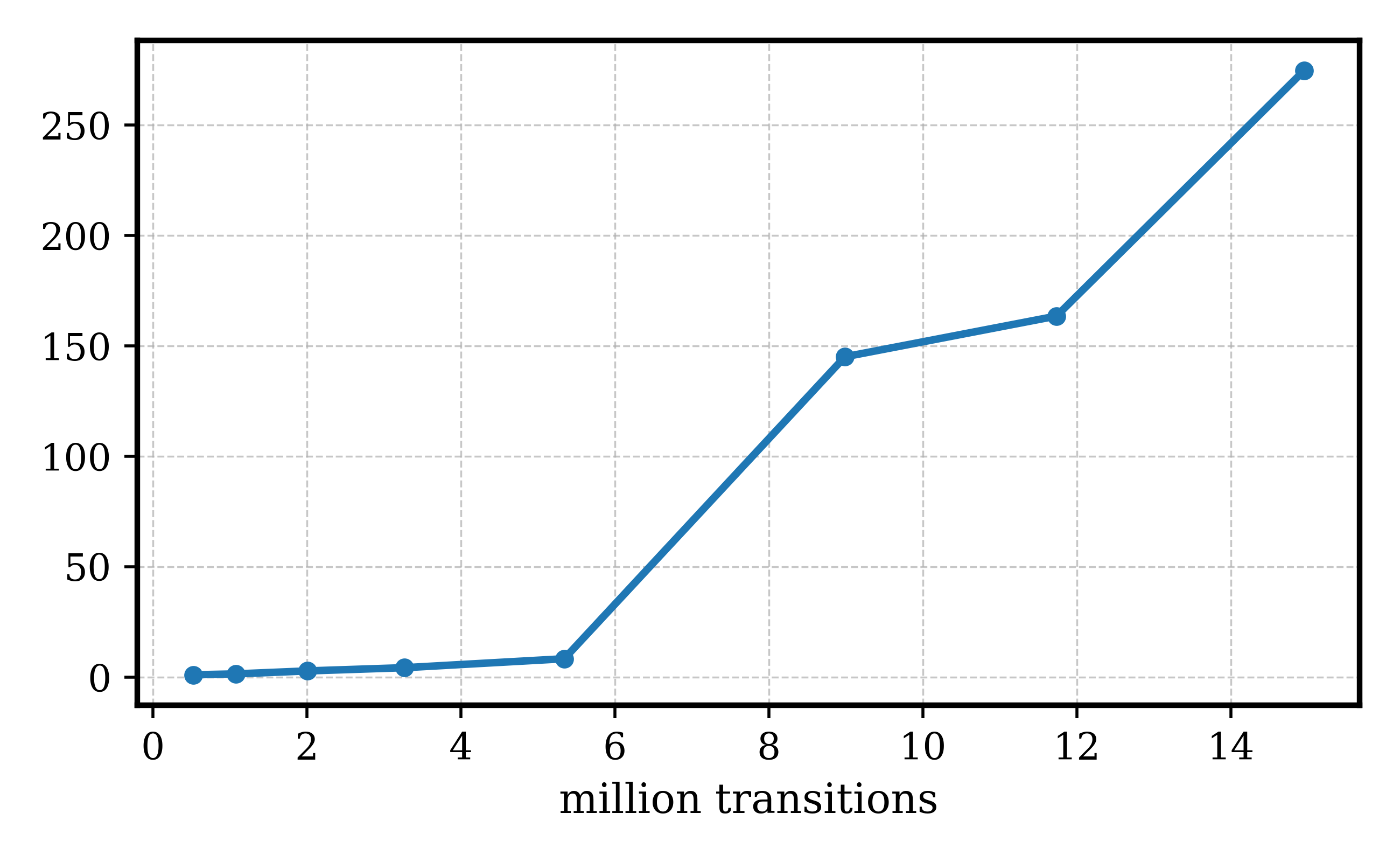}
    \caption{Qualitative analysis time (seconds) vs.\ number of transitions (millions) and states}
    \label{fig:Qual-Time-vs-criteria}
\end{figure}

\begin{figure}[h!]
\centering
\includegraphics[scale=0.45]{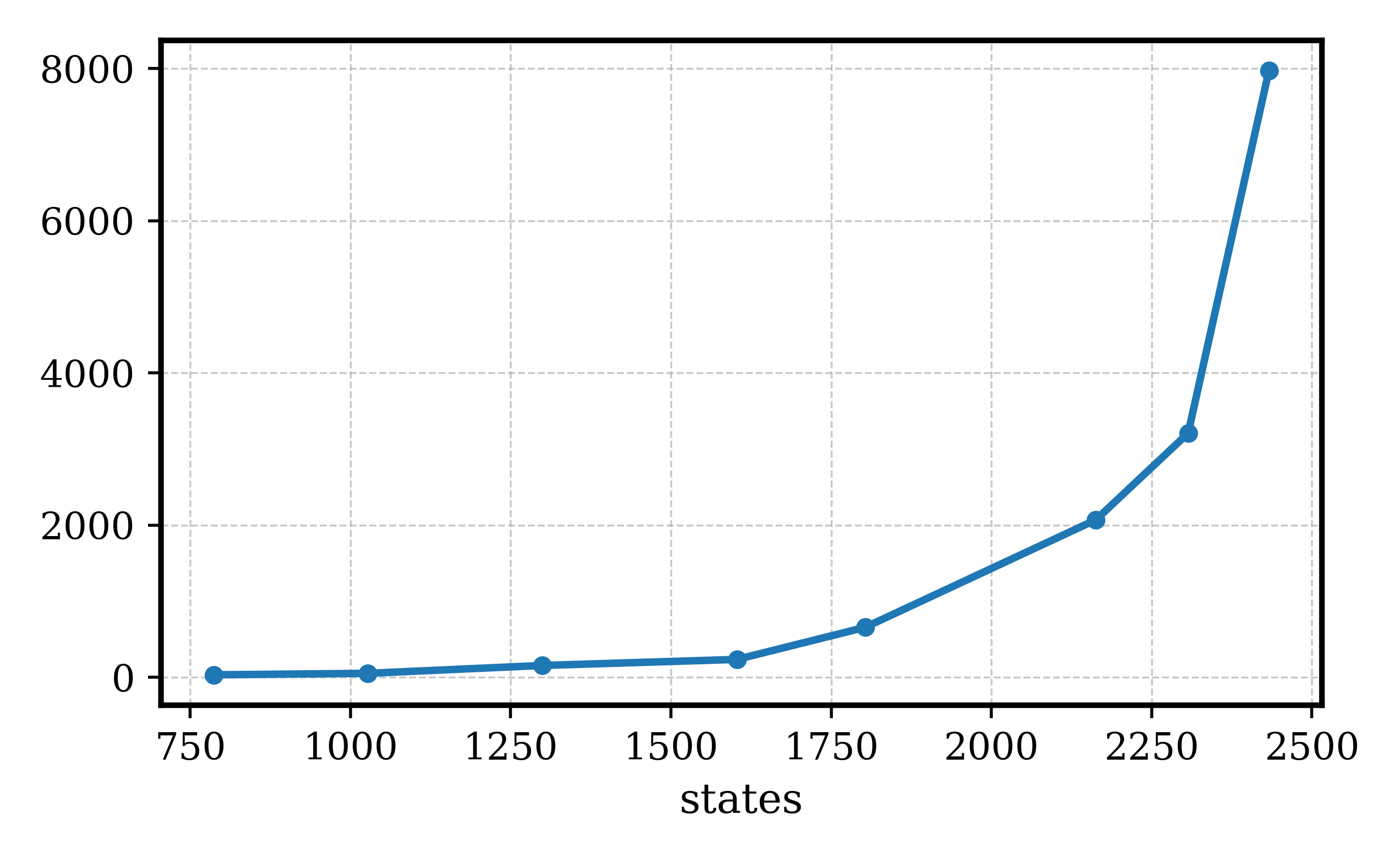}
    \includegraphics[scale=0.45]{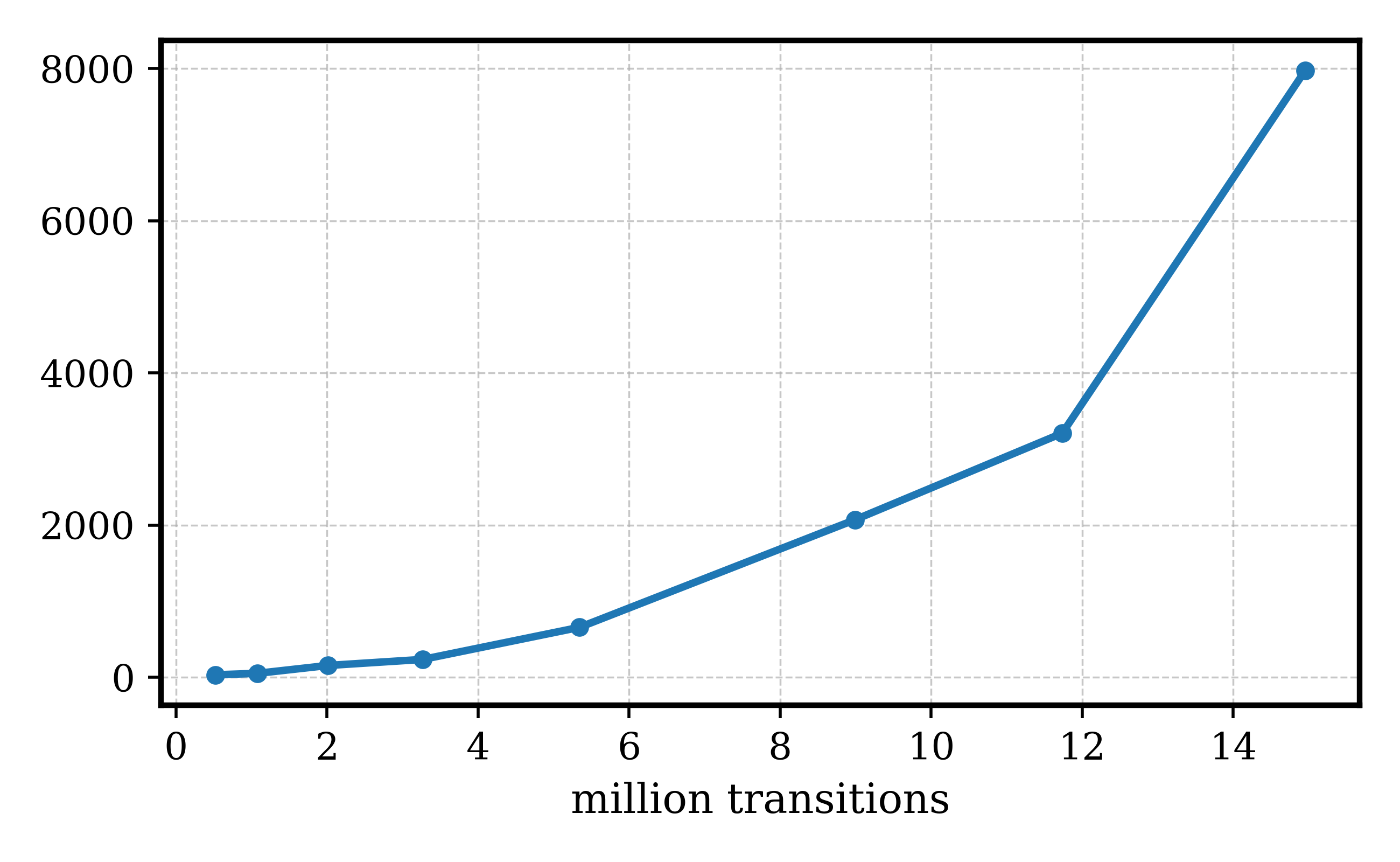}
    \caption{VI convergence time (seconds) vs.\ number of transitions (millions) and states}
    \label{fig:Val_iter_time_vs_criteria}
\end{figure}

\begin{figure}[h!]
\centering
    \includegraphics[scale=0.45]{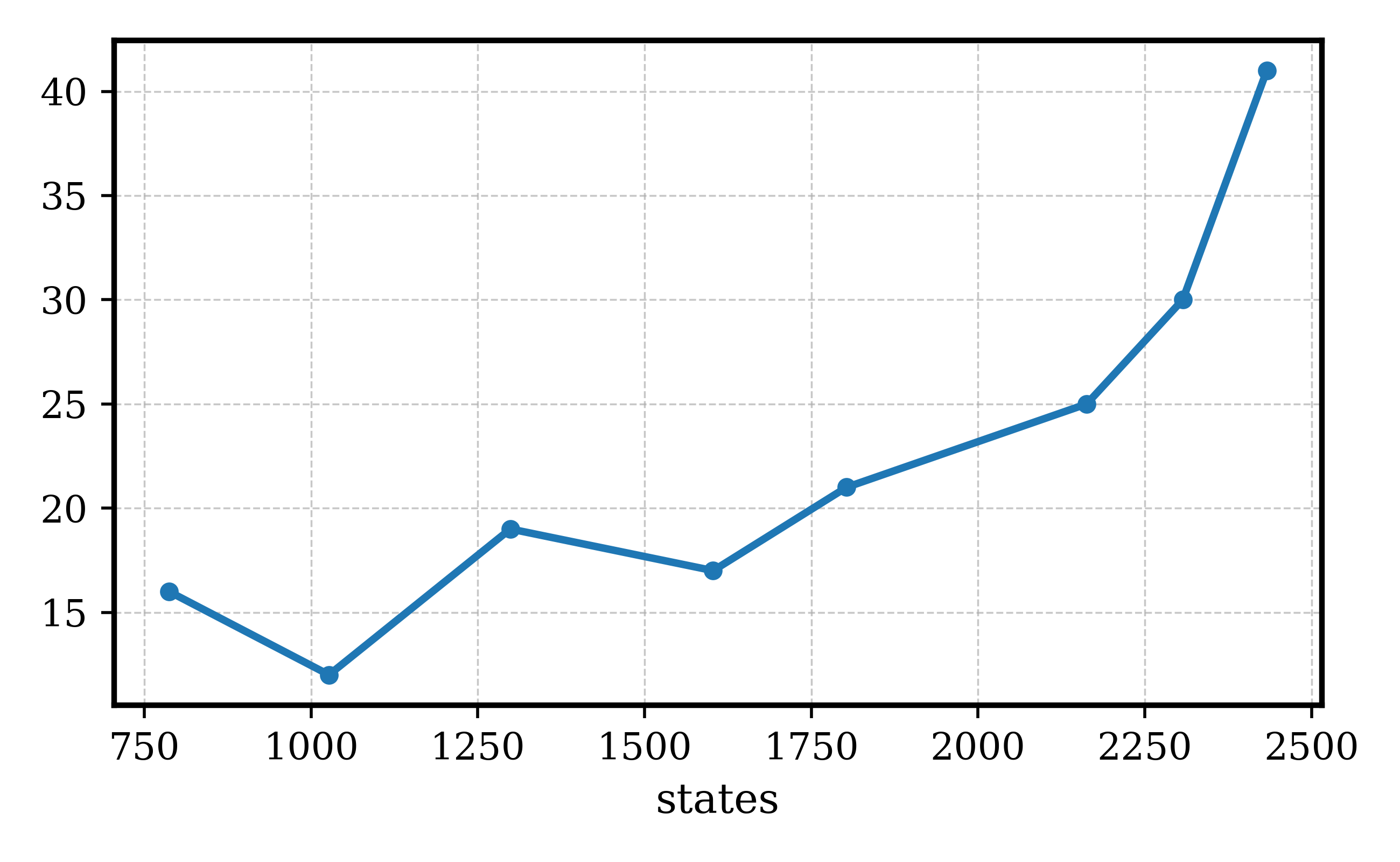}
    \includegraphics[scale=0.45]{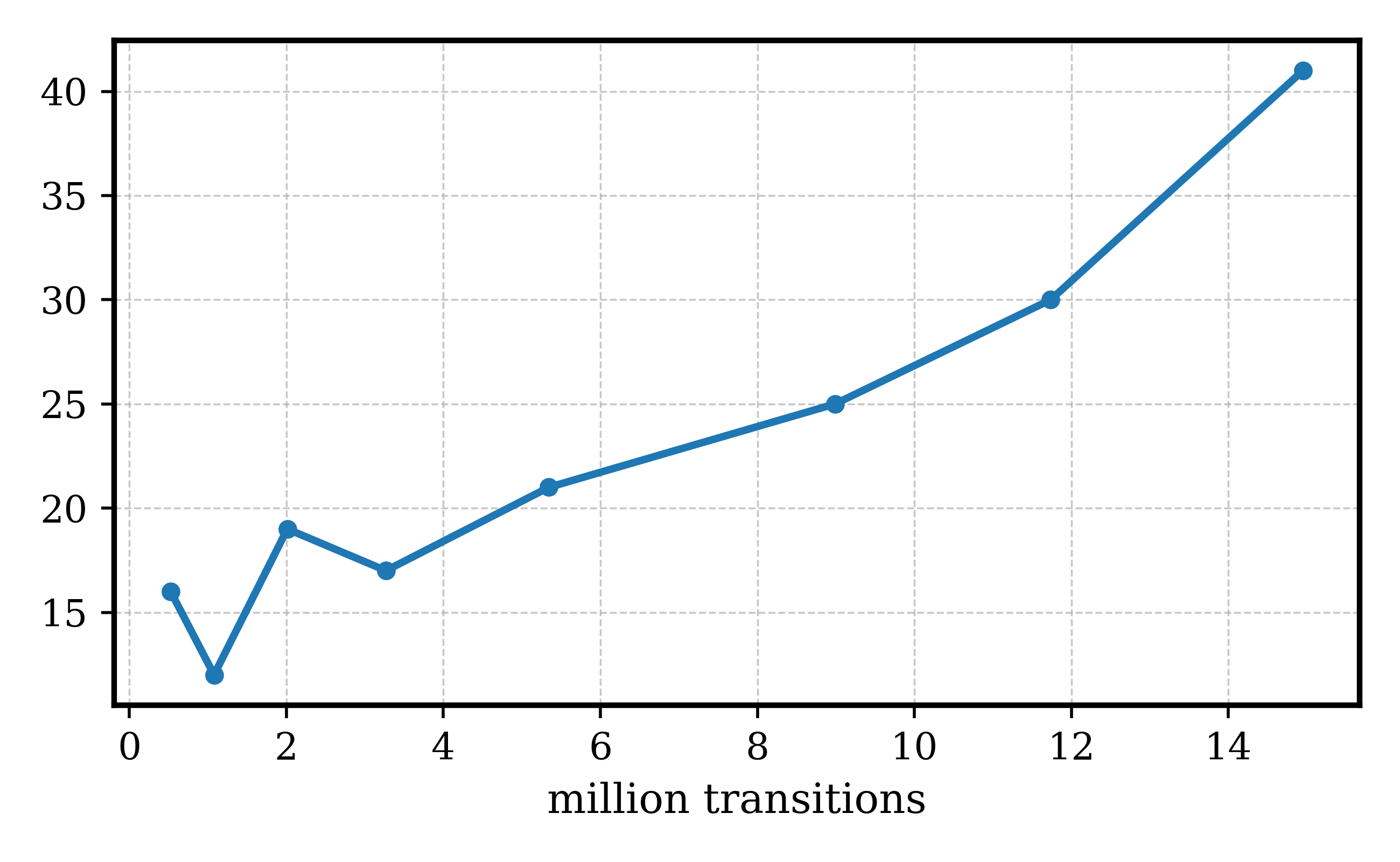}
    \caption{VI convergence iterations vs.\ number of transitions and states}
    \label{fig:Val_iter-cnvrg-iter-vs-criteria}
\end{figure}

% \begin{figure}[h!]
%         \centering
%         \includegraphics[width=.8\textwidth]{Figures/initial_states/Evolution_InitSt_VI_1600-3-withPI.png}
%         \caption{VI lower and upper values and SI values for SIMDP with 1600 states.}
%         \label{fig:app-val-iter-uperlower-init}
%     \end{figure}

\begin{center}
\begin{figure}[h!]
\centering
    \includegraphics[scale=0.45]{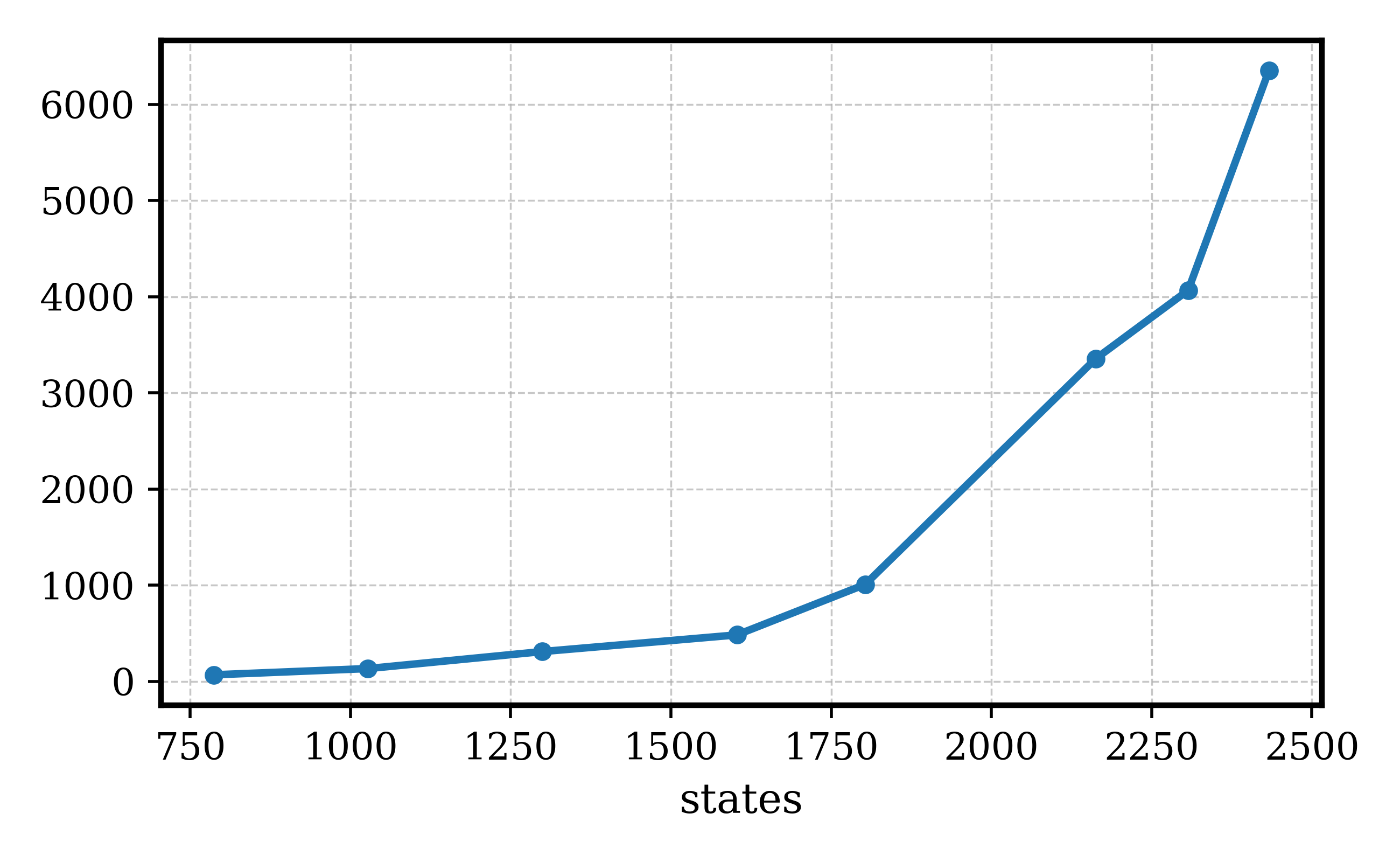}
    \includegraphics[scale=0.45]{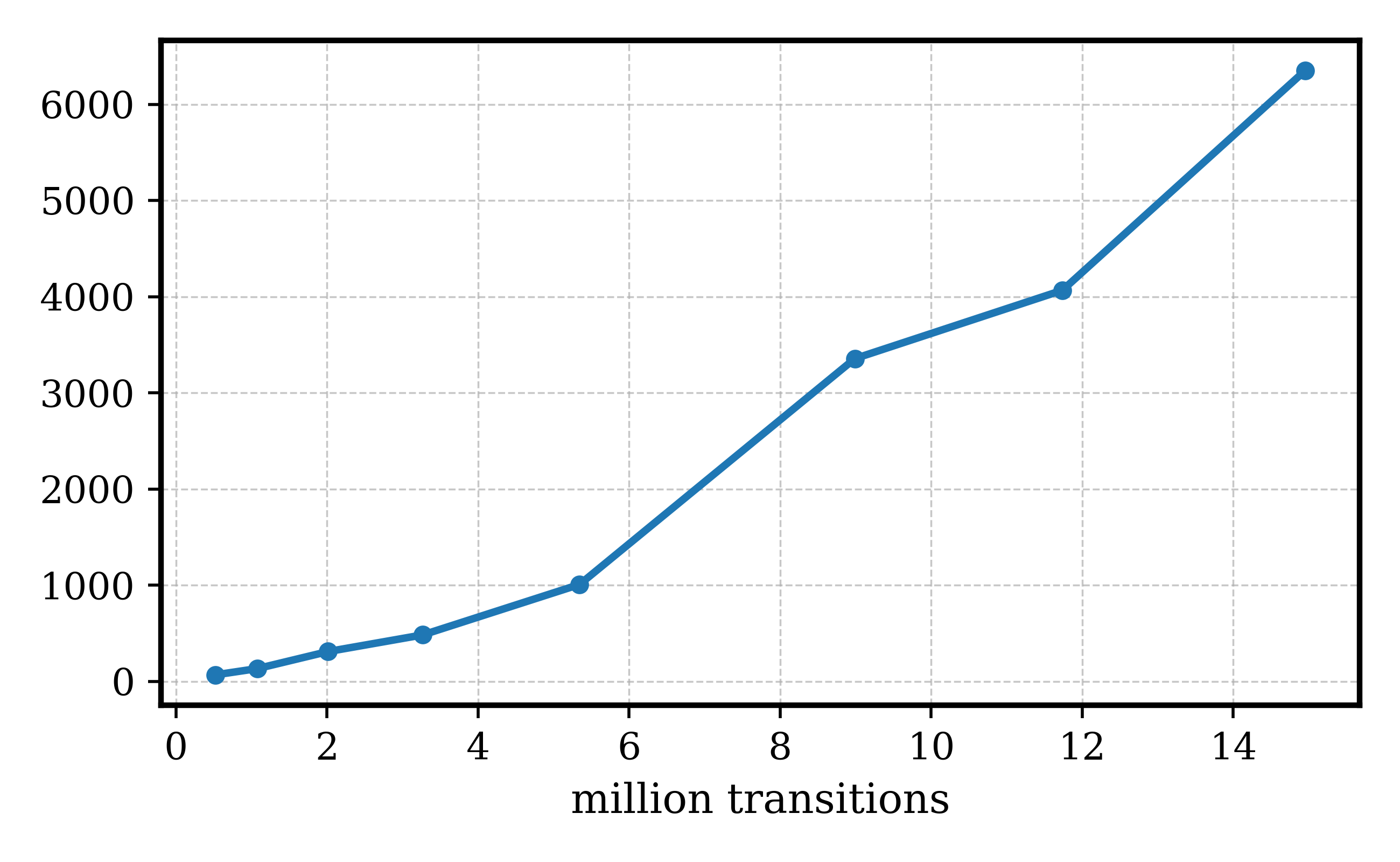}
    \caption{SI time in seconds vs.\ number of transitions and states}
    \label{fig:si-time-vs-criteria}
\end{figure}
\end{center}

\begin{figure}[h!]
    \centering
    \includegraphics[scale=0.55]{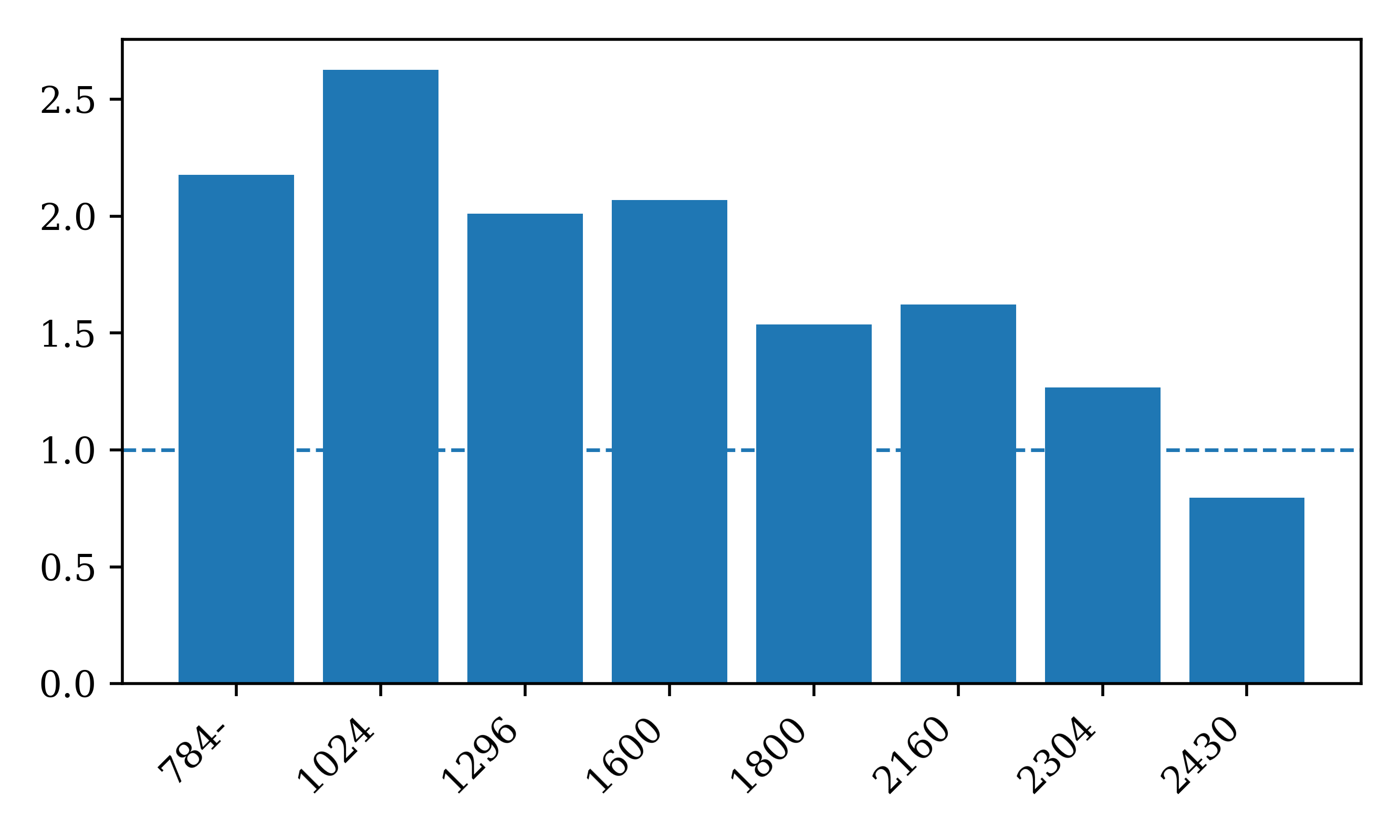}
    \caption{Ratio of SI time and VI time versus number of states.}
    \label{fig:SI-over-VI-ratio}
\end{figure}

\begin{figure}[h]
        \centering
        \includegraphics[width=\columnwidth,trim={0 0.3cm 0 2cm},clip]{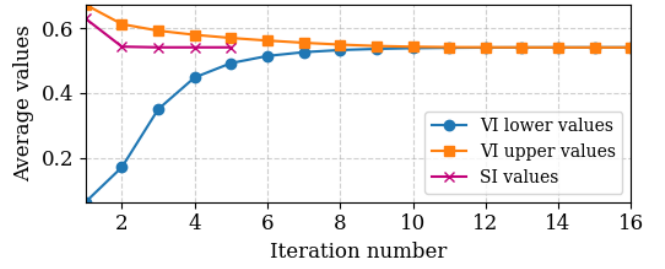}
        \centering
          \caption{Average initial-state values for VI lower bounds, VI upper bounds, and SI across iterations for an IMDP with 1600 states under $\mathbf{GF}(\textit{reach})$.}
        \label{fig:app-val-iter-uperlower-init}
    \end{figure}

\newpage
\subsection{Details on GR(1)}
In addition to general LTL specifications, we evaluate our framework on the GR(1) fragment, which provides a structured and practically useful class of reactive objectives. A GR(1) specification has the form
% \mbox{$
\[
\Bigl(\bigwedge_{i=1}^{m} \mathbf{GF}\,\varphi_i\Bigr)
\;\rightarrow\;
\Bigl(\bigwedge_{j=1}^{n} \mathbf{GF}\,\psi_j\Bigr),
\]
% $}
where the boolean formulas $\varphi_i$ and $\psi_i$ represent assumptions and guarantees, respectively: if the assumptions hold infinitely often, then the guarantees must hold infinitely often. GR(1) is substantially richer than plain reachability or a single B\"uchi objective, yet still structured enough to remain tractable. GR(1) specifications naturally induce nondeterministic automata. In the context of IMDPs, GR(1) allows us to express recurring liveness requirements together with safety-sensitive behavior, such as repeated recovery, repeated successful task completion, and avoidance of deadlocks.
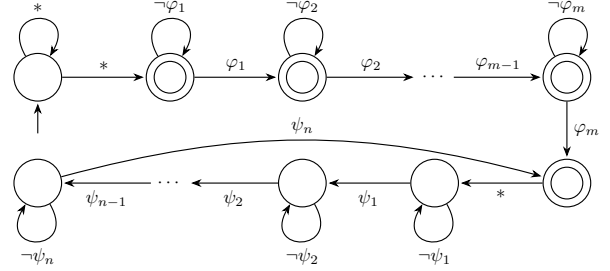
\begin{figure}[h]
\centering
\scalebox{0.7}{
\begin{tikzpicture}[
    >=Stealth,
    shorten >=1pt,
    node distance=2.5cm,
    on grid,
    auto,
    semithick,
    state/.style={circle,draw,minimum size=9mm},
    acc/.style={state, path picture={\draw circle[radius=3mm];}}]

% ---------- Top row ----------
\node[state] (q0) {};
% \node[state, initial] (q0) {};
\node[acc, right=of q0] (q1) {};
\node[acc, right=of q1] (q2) {};
\node[right=of q2] (dots1) {$\cdots$};
\node[acc, right=of dots1] (qm) {};

% ---------- Bottom row ----------
\node[acc, below=2cm of qm] (p0) {};
\node[state, left=of p0] (p1) {};
\node[state, left=of p1] (p2) {};
\node[left=of p2] (dots2) {$\cdots$};
\node[state, left=of dots2] (pm) {};

% ---------- Top row edges ----------

\draw[->] ($(q0.south)+(0,-0.6cm)$) -- (q0.south);

\path[->]
(q0) edge[loop above] node {$*$} ()
     edge node {$*$} (q1)

(q1) edge[loop above] node {$\neg \varphi_1$} ()
     edge node {$\varphi_1$} (q2)

(q2) edge[loop above] node {$\neg \varphi_2$} ()
     edge node {$\varphi_2$} (dots1)

(dots1) edge node {$\varphi_{m-1}$} (qm)

(qm) edge[loop above] node {$\neg \varphi_m$} ()
     edge node[right] {$\varphi_m$} (p0);

% ---------- Bottom row edges ----------
\path[->]
(p0) edge node {$*$} (p1)

(p1) edge[loop below] node {$\neg \psi_1$} ()
     edge node[below] {$\psi_1$} (p2)

(p2) edge[loop below] node {$\neg \psi_2$} ()
     edge node[below] {$\psi_2$} (dots2)

(dots2) edge node[below] {$\psi_{n-1}$} (pm)

(pm) edge[loop below] node {$\neg \psi_n$} ()
     edge[bend left=15] node[above] {$\psi_n$} (p0);

\end{tikzpicture}
}
\caption{GR(1) automaton.}
\label{fig:gr1_schema}
\end{figure}

% \begin{figure*}[b]
%     \centering    
%     \begin{subfigure}[b]{0.32\textwidth}
%         \centering
%         \includegraphics[width=\linewidth]{Figures/gr1_uav_si_time_vs_transitions.png}
%     \end{subfigure}
%     \hfill
%     \begin{subfigure}[b]{0.32\textwidth}
%         \centering
%         \includegraphics[width=\linewidth]{Figures/gr1_uav_vi_time_vs_transitions.png}
%     \end{subfigure}
%     \hfill
%         \begin{subfigure}[b]{0.32\textwidth}
%         \centering
%         \includegraphics[width=\linewidth]{Figures/gr1_uav_qual_time_vs_transitions.png}
%     \end{subfigure}

%     \caption{
% Performance comparison for the UAV benchmarks under the selected GR(1) specifications. The columns correspond to Strategy Improvement (left), Value Iteration (middle), and Qualitative Analysis (right). All times are reported in seconds. The legend is shown only once; the same color-to-specification mapping is used across all subfigures.
%     }
%     \label{fig:app-gr1_uav_results}
% \end{figure*}

\Cref{fig:gr1_schema} shows the automaton schema underlying GR(1) specifications: the upper part corresponds to the assumptions ($\varphi_i$) and the lower part corresponds to the guarantees ($\psi_i$). We then instantiate this pattern with the following four GR(1) specifications over the atomic propositions \textit{init}, \textit{reach} (\textit{r}), and \textit{deadlock} (\textit{d}):
\[
\begin{aligned}
\Phi_1:\quad
& \mathbf{GF}(init)\land \mathbf{GF}(\neg deadlock)
   \rightarrow \mathbf{GF}(reach) \\
\Phi_2:\quad
& \mathbf{GF}(\neg deadlock)
   \rightarrow \mathbf{GF}(reach)\land \mathbf{GF}(init) \\
\Phi_3:\quad
& \mathbf{GF}(\neg deadlock)
   \rightarrow \mathbf{GF}(reach \land \neg init) \\
\Phi_4:\quad & \mathbf{GF}(init \land \neg deadlock) \\ & \qquad\rightarrow \mathbf{GF}(reach \land \neg deadlock).
% \text{GR(1)}_4:\quad & \mathbf{GF}(init \land \neg deadlock) \rightarrow \mathbf{GF}(reach \land \neg deadlock)
\end{aligned}
\]
These specifications capture richer recurring mission requirements. 
 $\Phi_1$ requires eventual recurring success under recurring initialization and infinitely often absence of deadlock. The automaton modeling $\Phi_1$ can be found in Figure \ref{fig:gr1_1}.
 $\Phi_2$ models repeated recovery-and-success behavior, requiring the system to repeatedly revisit both \textit{init} (as recovery) and \textit{reach} (as success). (automaton in Figure \ref{fig:gr1_2}).
 $\Phi_3$ strengthens this by requiring nontrivial recurring success, namely reaching a progressed goal state distinct from the initial region. (automaton in Figure \ref{fig:gr1_3}). 
 $\Phi_4$ is deadlock-sensitive and requires that whenever the system can repeatedly return to a live initial condition, it must also repeatedly achieve a live success condition. automaton in Figure \ref{fig:gr1_4}).

% \begin{figure}[h]
% \centering
% \scalebox{0.7}{
% \input{Figures/gr1_main}
% }
% \caption{GR(1) automaton.}
% \label{fig:gr1_schema_1row}
% \end{figure} 

\begin{figure}[h]
\centering
\scalebox{0.7}{
\begin{tikzpicture}[
    >=Stealth,
    shorten >=1pt,
    node distance=3cm,
    on grid,
    auto,
    semithick,
    state/.style={circle,draw,minimum size=9mm},
    acc/.style={state, path picture={\draw circle[radius=3mm];}}] % adjust inner circle size here   

\node[state] (q0) {};
\node[acc, right=2cm of q0] (q1) {};
\node[acc, right=of q1] (q2) {};
\node[acc, right=of q2] (p0) {};
\node[state, right=2cm of p0] (p1) {};

\draw[->] ($(q0.south)+(0,-0.6cm)$) -- (q0.south);

\path[->]

(q0) edge[loop above] node {$*$} ()
     edge node {$*$} (q1)

(q1) edge[loop above] node {$\neg init$} ()
     edge node {$init$}  (q2)

(q2) edge[loop above] node {$deadlock$} ()
     edge node {$\neg deadlock$}  (p0);

\path[->]
(p0) edge node {$*$} (p1)

(p1) edge[loop above] node {$\neg reach$} ()
     edge[bend left=45] node[below] {$reach$} (p0);

\end{tikzpicture}
}
\caption{\mbox{$\Phi_1:$ $\mathbf{GF}(init)\land \mathbf{GF}(\neg d)\rightarrow \mathbf{GF}(r)$.}}
\label{fig:gr1_1}
\end{figure}

\begin{figure}[h]
\centering
\scalebox{0.7}{
\begin{tikzpicture}[
    >=Stealth,
    shorten >=1pt,
    node distance=3cm,
    on grid,
    auto,
    semithick,
    state/.style={circle,draw,minimum size=9mm},
    acc/.style={state, path picture={\draw circle[radius=3mm];}}]

\node[state] (q0) {};
\node[acc, right=2cm of q0] (q1) {};
\node[acc, right=of q1] (p0) {};
\node[state, right=2cm of p0] (p1) {};
\node[state, right=of p1] (p2) {};

\draw[->] ($(q0.south)+(0,-0.6cm)$) -- (q0.south);

\path[->]
(q0) edge[loop above] node {$*$} ()
     edge node {$*$} (q1)
(q1) edge[loop above] node {$\mathit{deadlock}$} ()
     edge node {$\neg \mathit{deadlock}$}  (p0);

\path[->]
(p0) edge node {$*$} (p1)
(p1) edge[loop above] node {$\neg \mathit{init}$} ()
     edge node {$\mathit{init}$} (p2)
(p2) edge[loop above] node {$\neg \mathit{reach}$} ()
     edge[bend left=45] node[below] {$\mathit{reach}$} (p0);

\end{tikzpicture}
}
\caption{\mbox{
$\Phi_2:$ $\mathbf{GF}(\neg d)
   \rightarrow \mathbf{GF}(r)\land \mathbf{GF}(init)$.}}  
\label{fig:gr1_2}
\end{figure}

\begin{figure}[h]
\centering
\scalebox{0.7}{
\begin{tikzpicture}[
    >=Stealth,
    shorten >=1pt,
    node distance=3.8cm,
    on grid,
    auto,
    semithick,
    state/.style={circle,draw,minimum size=9mm},
    acc/.style={state, path picture={\draw circle[radius=3mm];}}] % adjust inner circle size here   

% States
\node[state] (q0) {};
\node[acc, right=2cm of q0] (q1) {};
\node[acc, right=of q1] (p0) {};
\node[state, right=of p0] (p1) {};

% Left part
\draw[->] ($(q0.south)+(0,-0.6cm)$) -- (q0.south);

\path[->]
(q0) edge[loop above] node {$*$} ()
     edge node {$*$} (q1)

(q1) edge[loop above] node {$\mathit{deadlock}$} ()
     edge node {$\neg \mathit{deadlock}$}  (p0);

% Right part
\path[->]
(p0) edge node {$*$} (p1)

(p1) edge[loop above] node {$\neg \mathit{reach} \lor \mathit{init}$} ()
     edge[bend left=45] node[below] {$\mathit{reach \land \neg init}$} (p0);

\end{tikzpicture}}
\caption{\mbox{
$\Phi_3:$ $\mathbf{GF}(\neg d)
   \rightarrow \mathbf{GF}(r \land \neg init)$.
   }}  
\label{fig:gr1_3}
\end{figure} 

\begin{figure}[h]
\centering
\scalebox{0.7}{
\begin{tikzpicture}[
    >=Stealth,
    shorten >=1pt,
    node distance=3.8cm,
    on grid,
    auto,
    semithick,
    state/.style={circle,draw,minimum size=9mm},
    acc/.style={state, path picture={\draw circle[radius=3mm];}}] % adjust inner circle size here   

% States
\node[state] (q0) {};
% \node[state, initial] (q0) {};
\node[acc, right=2cm of q0] (q1) {};
\node[acc, right=of q1] (p0) {};
\node[state, right=of p0] (p1) {};

% Left part
\draw[->] ($(q0.south)+(0,-0.6cm)$) -- (q0.south);

\path[->]
(q0) edge[loop above] node {$*$} ()
     edge node {$*$} (q1)

(q1) edge[loop above] node {$\neg init \lor deadlock$} ()
     edge node {$init \land \neg deadlock$}  (p0);

% Right part
\path[->]
(p0) edge node {$*$} (p1)

(p1) edge[loop above] node {$\neg reach \lor deadlock$} ()
     edge[bend left=45] node[below] {$reach \land \neg deadlock$} (p0);

\end{tikzpicture}

% \begin{tikzpicture}[
%     >=Stealth,
%     shorten >=1pt,
%     node distance=2.8cm,
%     on grid,
%     auto,
%     semithick,
%     state/.style={circle,draw,minimum size=9mm},
%     acc/.style={state, path picture={\draw circle[radius=3mm];}}] % adjust inner circle size here   

% % States
% \node[state, initial] (q0) {};
% \node[acc, right=of q0] (q1) {};
% \node[acc, right=of q1] (q2) {};
% \node[acc, right=of q2] (p0) {};
% \node[state, right=of p0] (p1) {};

% % Left part
% \path[->]
% (q0) edge[loop above] node {$*$} ()
%      edge node {$*$} (q1)

% (q1) edge[loop above] node {$\neg init$} ()
%      edge node {$init$}  (q2)

% (q2) edge[loop above] node {$deadlock$} ()
%      edge node {$\neg deadlock$}  (p0);

% % Right part
% \path[->]
% (p0) edge (p1)

% (p1) edge[loop above] node {$\neg reached$} ()
%      edge[bend left=45] node[below] {$reached$} (p0);

% \end{tikzpicture}
}
\caption{\mbox{
$\Phi_4:$ $\mathbf{GF}(init \land \neg d)\rightarrow \mathbf{GF}(r \land \neg d)$.
   }}  
\label{fig:gr1_4}
\end{figure}

Figure~\ref{fig:app-gr1_uav_results} reports the runtime for SI, VI, and qualitative analysis for the UAV benchmarks as a function of the number of transitions. These experiments show that our framework handles GR(1)-induced nondeterminism in large IMDPs and can support richer reactive mission requirements than simple reachability-style specifications.

\begin{figure}[h]
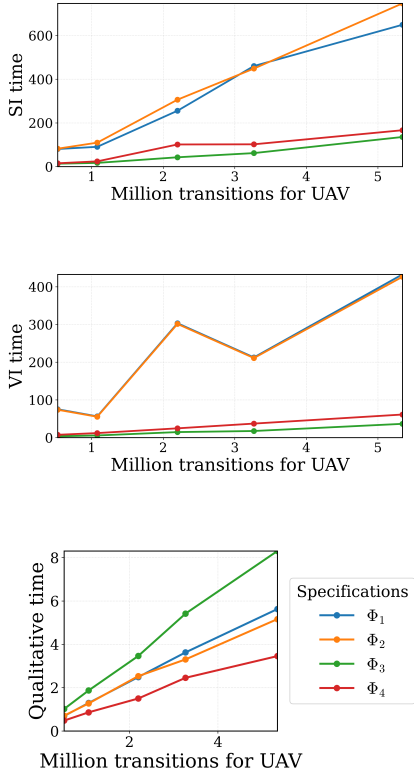

    \centering    
    \begin{subfigure}[b]{0.32\textwidth}
        \centering
        \includegraphics[width=\linewidth]{Figures/gr1_uav_si_time_vs_transitions.png}
    \end{subfigure}
    % \hfill
    \begin{subfigure}[b]{0.32\textwidth}
        \centering
        \includegraphics[width=\linewidth]{Figures/gr1_uav_vi_time_vs_transitions.png}
    \end{subfigure}
    % \hfill
        \begin{subfigure}[b]{0.32\textwidth}
        % \centering
        \raggedleft

        \includegraphics[width=1.1\linewidth]{Figures/gr1_uav_qual_time_vs_transitions.png}
    \end{subfigure}

    \caption{
Performance comparison for the UAV benchmarks under the selected GR(1) specifications. The plots correspond to Strategy Improvement (top), Value Iteration (middle), and Qualitative Analysis (bottom). All times are reported in seconds. The legend is shown only once; the same color-to-specification mapping is used across all subfigures.
    }
    \label{fig:app-gr1_uav_results}
\end{figure}

\begin{comment}
\begin{figure*}[h]
    \centering
    \begin{subfigure}[b]{0.32\textwidth}
        \centering
        \includegraphics[width=\linewidth]{Figures/gr1_uav_si_time_vs_transitions.png}
    \end{subfigure}\hspace{0.01\textwidth}
    \begin{subfigure}[b]{0.32\textwidth}
        \centering
        \includegraphics[width=\linewidth]{Figures/gr1_uav_vi_time_vs_transitions.png}
    \end{subfigure}\hspace{0.01\textwidth}
    \begin{subfigure}[b]{0.32\textwidth}
        \centering
        \includegraphics[width=\linewidth]{Figures/gr1_uav_qual_time_vs_transitions.png}
    \end{subfigure}
\end{figure*}

\end{comment}

\end{document}